\documentclass[draftclsnofoot,peerreview,english]{IEEEtran}
\usepackage[T1]{fontenc}
\usepackage[latin9]{inputenc}
\usepackage{amsmath}
\usepackage{amssymb}
\usepackage{mathrsfs}
\usepackage{bbm}
\usepackage{mathtools}
\usepackage{empheq}
\usepackage{graphicx}
\usepackage{amsthm}	
\usepackage{hyperref}
\usepackage{cite}


\usepackage{babel}

\newtheorem{thm}{Theorem}				%instance of the pkg amsthm
\newtheorem{prop}[thm]{Proposition}		%treat \prop like \thm
\newtheorem{lem}[thm]{Lemma}			%[thm] means lems and thms share same num seq
\newtheorem{cor}[thm]{Corollary}

\DeclareMathOperator{\pr}{\mathrm{Pr}} 		%probability

\begin{document}

\title{Can Negligible Cooperation\\ Increase Network Reliability?}

\author{Parham Noorzad, \IEEEmembership{Student Member, IEEE,} 
Michelle Effros, \IEEEmembership{Fellow, IEEE,}\\ and Michael
Langberg, \IEEEmembership{Senior Member, IEEE}%
\thanks{This paper was presented in part at the 2016 IEEE International 
Symposium of Information Theory in Barcelona, Spain.}
\thanks{This material is based upon work supported by the 
National Science Foundation under Grant Numbers 15727524,
1526771, and 1321129. }%
\thanks{P. Noorzad and M. Effros are with the California
Institute of Technology, Pasadena, CA 91125 USA 
(emails: parham@caltech.edu, effros@caltech.edu). }
\thanks{M. Langberg is with the State
University of New York at Buffalo, Buffalo, NY 14260 USA
(email: mikel@buffalo.edu).}}

\maketitle

\begin{abstract} 
In network cooperation strategies, nodes work together with the aim of 
increasing transmission rates or reliability. 
This paper demonstrates that enabling cooperation 
between the transmitters of a two-user multiple access channel, via a cooperation 
facilitator that has access to both messages, 
always results in a network whose maximal- and average-error 
sum-capacities are the same---even when those capacities differ in the 
absence of cooperation and the information shared with the encoders is negligible. 
From this result, it follows that if a multiple access channel with 
no transmitter cooperation has different maximal- and average-error sum-capacities, 
then the maximal-error sum-capacity of the network consisting of this channel
and a cooperation facilitator is not continuous with respect to the 
output edge capacities of the facilitator. This shows that there exist networks where 
sharing even a negligible number of bits per channel use with the encoders
yields a non-negligible benefit. 
\end{abstract}


\section{Introduction} \label{sec:intro}

In his seminal work \cite{Shannon}, Shannon defines capacity as the maximal rate
achievable with arbitrarily small {\em maximal error probability}.  For
the point-to-point channel studied in that work, the rates achievable with
arbitrarily small {\em average error probability} turn out to be the same.
Since average error probability is often easier to work with but maximal
error probility provides a more stringent constraint that for some
channels yields a different capacity region~\cite{Dueck}, both maximal and
average error probability persist in the literature.

For multiterminal channels, the benefit of codes with small maximal error probability 
may come at the cost of lower rates. For example, in the MAC, the sum rates achievable
under maximal and average error probability constraints can differ \cite{Dueck}.
While such differences cannot arise in the broadcast channel \cite{WillemsBC},
the MAC with full encoder cooperation\footnote{In a 2-user MAC with full encoder cooperation, both 
encoders have access to both messages.}, or other scenarios that actually
or effectively employ only a single encoder, the importance of networks with
multiple encoders and infeasibility of full cooperation in many scenarios, together
motivate our interest in quantifying the reliability benefit of \emph{rate-limited}
cooperation.  

To make this discussion concrete, consider a network consisting of a 
multiple access channel (MAC) and a cooperation facilitator (CF)
\cite{NoorzadEtAl,Noorzad2}, as shown in Figure \ref{fig:network}. The CF is a node
that sends and receives limited information to and from each encoder. 
Prior to transmitting its codeword over the channel, 
each encoder sends some information to the CF. The CF then replies
to each encoder over its output links. This communication may continue for a finite 
number of rounds. The total number of bits transmitted on each CF input or output link is bounded by 
the product of the blocklength, $n$, and the capacity of that link. Once the encoders'
communication with the CF is over, each encoder transmits its codeword 
over $n$ channel uses. 

In order to quantify the benefit of rate-limited cooperation in the above network, 
we define a spectrum of error probabilities that range from average error to maximal error.
Theorem \ref{thm:main}, decribed in Subsection \ref{subsec:coopRel}, 
states that if for $i\in\{1,2\}$, we increase $C^i_\mathrm{in}$ (the capacity
of the link from encoder $i$ to the CF) by a value proportional to 
the desired increase in reliability, and $C^i_\mathrm{out}$ (the capacity
of the link from the CF to encoder $i$) by any arbitrarily small amount, then
any rate pair that is achievable in the original network under average 
error is achievable in the new network under a stricter notion of error. This
result quantifies the relationship between cooperation under the CF model and reliability.
For the proof, we use techniques from \cite{WillemsBC}, in which Willems shows
that the average- and maximal-error capacity regions of the discrete
memoryless broadcast channel are identical. A similar result, quantifying the 
reliability benefit of cooperation under the conferencing encoders model \cite{WillemsMAC},
appears in Subsection \ref{subsec:confRel}. 
\begin{figure} 
	\begin{center}
		\includegraphics[scale=0.25]{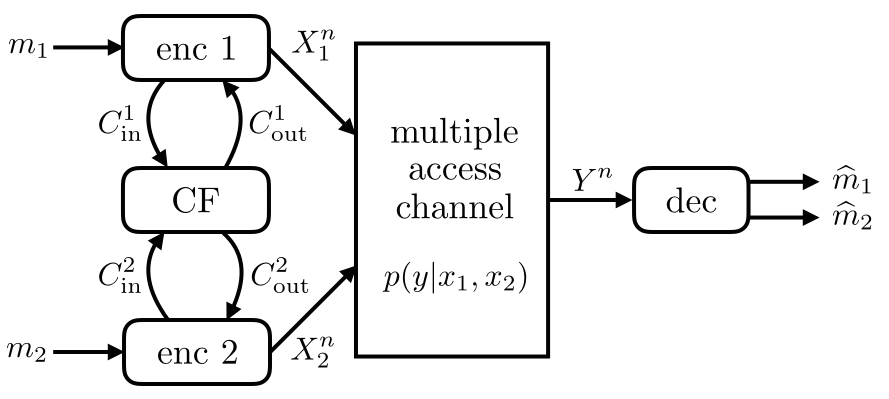}
		\caption{A network consisting of a memoryless MAC
		and a CF.} \label{fig:network}
	\end{center}
\end{figure} 

Our main result, Theorem \ref{thm:MaxEqAvg}, considers the case
where $C_\mathrm{in}^1$ and $C_\mathrm{in}^2$ 
are sufficiently large so that the CF has access to both source messages. 
In such a network, whenever $C_\mathrm{out}^1$ and $C_\mathrm{out}^2$ 
are positive, the maximal- and average-error capacity regions are equal.
Thus, unlike the classical MAC scenario, where
codes with small maximal error achieve lower rates than codes with small average error, 
when the encoders cooperate through a CF that has full access to the messages
and outgoing links of arbitrarily small capacity, any rate pair that is achievable with 
small average error is also achievable with small maximal error. Therefore,
cooperation removes the tradeoff that exists between transmission rates and reliability
in the classical MAC. Thus in our model of rate-limited cooperation, it is 
possible to obtain the rate and reliability benefits of cooperation 
at the same time. 

Applying the equality between maximal- and average-error capacity regions in
the scenario described above to Dueck's ``Contraction MAC,'' a MAC with
maximal-error capacity region strictly smaller than its
average-error region \cite{Dueck}, yields a network whose
maximal-error sum-capacity is not continuous with respect to the capacities 
of its edges (Proposition \ref{prop:discontinuity}). 
The discontinuity in sum-capacity observed here is related to
the edge removal problem \cite{HoEtAl, JalaliEtAl}, which we next discuss. 

The edge removal problem studies the change in the capacity region of a network that
results from removing a point-to-point channel of finite capacity, here called an ``edge,'' 
from the network. One instance of this problem considers removed 
edges of ``negligible capacity.'' Intuitively, an
edge has negligible capacity if for all functions $f(n)=o(n)$ 
and all sufficiently large $n$, it can carry $f(n)$ bits 
noiselessly over the channel in $n$ channel uses. In this context, the edge removal 
problem asks whether removing an edge with negligible capacity from a network has any
effect on the capacity region of that network. Our result showing the existence of
a network with a discontinuous maximal-error sum-capacity demonstrates the existence of a 
network where removing an edge with negligible capacity has a non-negligible effect on 
its maximal-error capacity region (Subsection \ref{subsec:negligible}).

The edge removal problem for edges with negligible capacity has been studied previously.
In the context of lossless source coding over networks, 
Gu, Effros, and Bakshi \cite{GuEffrosBakshi} state the ``Vanishment Conjecture,'' which
roughly says that in a class of network source coding problems,
certain edges with negligible capacity can be removed with
no consequences. In \cite{GuEffros} and \cite[p. 51]{Gu}, the authors study the relation between
the edge removal problem for edges with negligible capacity and a notion of strong converse. 
In \cite{SarwateGastpar}, Sarwate and Gastpar show that feedback
via links of negligible capacity does not affect the average-error 
capacity region of a memoryless MAC.
In \cite{LangbergEffros1}, Langberg and Effros
demonstrate a connection between the edge removal problem for edges with negligible capacity and 
the equivalence between zero-error and $\epsilon$-error capacities in network coding. 
In recent work \cite{SingleBit}, Langberg and Effros show the existence of a network 
where even a single bit of communication results in a strictly larger maximal-error capacity region. 

Given that one may view feedback as a form of cooperation, similar questions 
may be posed about feedback and reliability. In \cite{Dueck}, Dueck 
shows that for some MACs, the maximal-error capacity region with feedback 
is strictly contained in the average-error region without feedback. 
This contrasts with our results on encoder cooperation via 
a CF that has access to both messages and output edges of negligible capacity.
Specifically, we show in Subsection \ref{subsec:negligible} that the maximal-error region 
of a MAC with negligible encoder cooperation of this kind contains 
the average-error region of the same MAC without encoder cooperation. 
For further discussion of results regarding feedback and the average- and
maximal-error regions of a MAC, we refer the reader to Cai \cite{Cai}. 

Other models under which maximal- and average-error capacity regions are 
identical include networks where one of the MAC encoders is ``stochastic,'' 
that is, its codewords 
depend on some randomly generated key in addition to its message. 
For such codes, the definitions of the maximal- and average-error probabilities 
require an expectation with respect to the distribution of the random bits. Cai shows in
\cite{Cai} that the maximal-error capacity region of a MAC where one encoder has access to a random
key of negligible rate equals the average-error capacity region of the same MAC when both 
encoders are deterministic. While some of the techniques we use in this paper are 
conceptually similar to Cai's proof
\cite{Cai}, the respective models are rather different. 
For example, it holds that stochastic encoders cannot
achieve higher rates than deterministic encoders under average error, 
even if they have access to random keys with positive rates. The same result
however, is not true of the cooperation model we study here when the cooperation
rate is positive \cite{kUserMACLong}. That is, at least for some MACs, a positive cooperation 
rate leads to a strictly positive gain. 
Furthermore, even for a negligible cooperation rate,
while we do not demonstrate a gain in the average-error capacity region, 
we are not able to rule out such a 
gain using the same proof that applies in the case of stochastic encoders.  

In the next section, we formally introduce our model. 
A discussion of our results follows in Section \ref{sec:results}. 

\section{Model} \label{sec:model}

Consider a network comprising two encoders, a 
cooperation facilitator (CF), a MAC
\begin{equation*}
  \big(\mathcal{X}_1\times\mathcal{X}_2,
	p(y|x_1,x_2),\mathcal{Y}\big),
\end{equation*}
and a decoder as depicted in Figure \ref{fig:network}. A CF is a node that communicates with the encoders
prior to the transmission of the codewords over the channel. This 
communication is made possible through noiseless links of capacities
$C_\mathrm{in}^1$ and $C_\mathrm{in}^2$ going from the CF to the encoders
and noiseless links of capacities $C_\mathrm{out}^1$ and $C_\mathrm{out}^2$ going 
back. 

Here our MAC may be discrete or continuous. In a discrete
MAC, the alphabets $\mathcal{X}_1$, $\mathcal{X}_2$, and $\mathcal{Y}$
are either finite or countably infinite, and for each $(x_1,x_2)$,
$p(y|x_1,x_2)$ is a probability mass function on $\mathcal{Y}$. In a continuous
MAC, $\mathcal{X}_1=\mathcal{X}_2=\mathcal{Y}=\mathbb{R}$, and 
$p(y|x_1,x_2)$ is a probability density function for each $(x_1,x_2)$. 
Furthermore, our MAC is memoryless and without feedback \cite[p. 193]{CoverThomas}, 
so for every positive integer $n$, the $n$th extension of our MAC is given 
by $(\mathcal{X}_1^n\times\mathcal{X}_2^n,p(y^n|x_1^n,x_2^n),\mathcal{Y}^n)$,
where
\begin{equation*}
  p(y^n|x_1^n,x_2^n)=\prod_{t=1}^n
	p(y_t|x_{1t},x_{2t}).
\end{equation*}

The following definitions aid our description of an $(n,M_1,M_2,J)$-code 
with transmitter cooperation for this network.
For every real number $x\geq 1$, let $[x]$ denote the set $\{1,\dots,\lfloor x\rfloor\}$,
where $\lfloor x\rfloor$ denotes the integer part of $x$. For each $i\in\{1,2\}$, fix 
two sequences of sets $(\mathcal{U}_{ij})_{j=1}^J$ and $(\mathcal{V}_{ij})_{j=1}^J$ 
such that 
\begin{align*}
  \log\big|\mathcal{U}_i^J\big|&=\sum_{j=1}^J\log|\mathcal{U}_{ij}|
  \leq nC_\mathrm{in}^i\\
  \log\big|\mathcal{V}_i^J\big|&=\sum_{j=1}^J\log|\mathcal{V}_{ij}|
  \leq nC_\mathrm{out}^i,
\end{align*}
where for all $j\in [J]$,
\begin{align*}
\mathcal{U}_i^j &=\prod_{\ell=1}^j\mathcal{U}_{i\ell}\\
\mathcal{V}_i^j &=\prod_{\ell=1}^j\mathcal{V}_{i\ell}, 
\end{align*}
and $\log$ denotes the logarithm base 2. Here $\mathcal{U}_{ij}$ represents
the alphabet for the round-$j$ transmission from encoder $i$ to the 
CF while $\mathcal{V}_{ij}$ represents the alphabet for the
round-$j$ transmission from the CF to encoder $i$. The given
alphabet size constraints are chosen to match the total rate constraints $nC_\mathrm{in}^i$
and $nC_\mathrm{out}^i$ over $J$ rounds of communication between the two 
encoders and $n$ uses of the channel. 
For $i\in\{1,2\}$, encoder $i$ is represented by $((\varphi_{ij})_{j=1}^J,f_i)$, 
where
\begin{equation*}
  \varphi_{ij} :[M_i]\times\mathcal{V}_i^{j-1} \rightarrow \mathcal{U}_{ij}
\end{equation*}
captures the round-$j$ transmission from encoder $i$ to the CF, and
\begin{equation*}
  f_i : [M_i]\times \mathcal{V}_i^J\rightarrow \mathcal{X}_i^n
\end{equation*}
captures the transmission of encoder $i$ across the channel.\footnote{Our 
results continue to hold for the case where the encoders 
satisfy individual cost constraints, since the same proofs 
apply with no modification.}
The CF is represented by the functions
$((\psi_{1j})_{j=1}^J,(\psi_{2j})_{j=1}^J)$, where
for $i\in\{1,2\}$ and $j\in [J]$,
\begin{equation*}
  \psi_{ij}:\mathcal{U}_1^j\times\mathcal{U}_2^j
  \rightarrow \mathcal{V}_{ij}
\end{equation*}
captures the round-$j$ transmission from the CF to encoder $i$.
For each message pair $(m_1,m_2)$ and $i\in\{1,2\}$,
define the sequences $(u_{ij})_{j\in [J]}$ and $(v_{ij})_{j\in [J]}$
recursively as 
\begin{align}
  u_{ij} &= \varphi_{ij}(m_i,v_i^{j-1}) \label{eq:uij}\\
  v_{ij} &= \psi_{ij}(u_1^j,u_2^j). \label{eq:vij}
\end{align}
In round $j$, encoder $i$ sends $u_{ij}$ to the CF
and receives $v_{ij}$ from the CF. After the $J$ round
communication between the encoders and the CF is over, encoder
$i$ transmits $f_i(m_i,v_i^J)$ over the channel. 
The decoder is represented by the function
\begin{equation*}
  g:\mathcal{Y}^n\rightarrow [M_1]\times [M_2].
\end{equation*}
The probability that a message pair $(m_1,m_2)$ is
decoded incorrectly is given by
\begin{equation*}
  \lambda_n(m_1,m_2)=\sum_{y^n\notin g^{-1}(m_1,m_2)}
  p\big(y^n|f_1(m_1,v_1^J),f_2(m_2,v_2^J)\big),
\end{equation*}
where
\begin{equation*}
  g^{-1}(m_1,m_2)
	=\big\{y^n\big|g(y^n)=(m_1,m_2)\big\}.
\end{equation*}
Note that $\lambda_n$ depends only on $(m_1,m_2)$ since by (\ref{eq:uij})
and (\ref{eq:vij}), $v_1^J$ and 
$v_2^J$ are deterministic functions of $(m_1,m_2)$. 
The average probability of error, $P_{e,\mathrm{avg}}^{(n)}$, and the 
maximal probability of error, $P_{e,\mathrm{max}}^{(n)}$, are defined as
\begin{align*}
  P_{e,\mathrm{avg}}^{(n)}
  &=\frac{1}{M_1M_2}\sum_{m_1,m_2}\lambda_n(m_1,m_2)\\
  P_{e,\mathrm{max}}^{(n)}
  &=\max_{m_1,m_2}\lambda_n(m_1,m_2),
\end{align*}
respectively. To quantify the reliability benefit of rate-limited 
cooperation, we require a more general notion of probability of 
error, which we next describe. 

For $r_1,r_2\geq 0$, the $(r_1,r_2)$-error probability $P_e^{(n)}(r_1,r_2)$
is a compromise between average and maximal error probability. 
To compute $P_e^{(n)}(r_1,r_2)$, we partition the matrix
\begin{equation} \label{eq:Lambda}
  \Lambda_n:=\big(\lambda_n(m_1,m_2)\big)_{m_1,m_2}
\end{equation}
into $K_1 K_2$ blocks of size $L_1\times L_2$, where for $i\in \{1,2\}$, 
\begin{align*}
  K_i &=\min\big\{\lfloor 2^{nr_i}\rfloor,M_i\big\}\\
  L_i &=\lfloor M_i/K_i\rfloor,
\end{align*}
and a single block containing the remaining $M_1M_2-K_1K_2L_1L_2$ 
entries. 
We begin by calculating the \emph{average} of the entries within each 
$L_1\times L_2$ block and obtain the $K_1K_2$ values 
\begin{equation*}
  \Bigg\{\frac{1}{L_1L_2}
  \sum_{\substack{m_1\in S_{1,k_1}\\ m_2\in S_{2,k_2}}}
  \lambda_n(m_1,m_2)\Bigg\}_{(k_1,k_2)},
\end{equation*}
where for $i\in\{1,2\}$ and $k_i\in [K_i]$,
the set $S_{i,k_i}\subseteq [M_i]$ is defined as
\begin{equation} \label{eq:Sk}
  S_{i,k_i}=\Big\{(k_i-1)L_i+1,\dots,k_i L_i\Big\}.
\end{equation}
Next we find the \emph{maximum} of the $K_1 K_2$ obtained average values,
namely 
\begin{equation} \label{eq:notPe}
  \max_{k_1,k_2}\frac{1}{L_1L_2}
  \sum_{\substack{m_1\in S_{1,k_1}\\ m_2\in S_{2,k_2}}}
  \lambda_n(m_1,m_2).
\end{equation}
The maximum in (\ref{eq:notPe}) depends on the labeling of the messages, which is
not desirable. To avoid this issue, we calculate the minimum of (\ref{eq:notPe})
over all permutations of the rows and columns of $\Lambda_n$. This results in the definition
\begin{equation*}
  P_e^{(n)}(r_1,r_2)=\min_{\pi_1,\pi_2}\max_{k_1,k_2}\frac{1}{L_1L_2}
  \sum_{\substack{m_1\in S_{1,k_1}\\ m_2\in S_{2,k_2}}}
  \lambda_n(\pi_1(m_1),\pi_2(m_2)),
\end{equation*}
where the minimum is over all permutations $\pi_1$ and $\pi_2$ of the sets $[M_1]$
and $[M_2]$, respectively. Note that $P_{e,\mathrm{avg}}^{(n)}$ and
$P_{e,\mathrm{max}}^{(n)}$ are special cases of $P_e^{(n)}(r_1,r_2)$, since
\begin{equation*}
  P_e^{(n)}(0,0) = P_{e,\mathrm{avg}}^{(n)},
\end{equation*}
and for sufficiently large values of $r_1$ and $r_2$, 
\begin{equation*}
  P_e^{(n)}(r_1,r_2) = P_{e,\mathrm{max}}^{(n)}.
\end{equation*}
 
We say a rate pair $(R_1,R_2)$ is $(r_1,r_2)$-error achievable
for a MAC with a 
$(\mathbf{C}_\mathrm{in},\mathbf{C}_\mathrm{out})$-CF and $J$ rounds of
cooperation if for all $\epsilon,\delta>0$,
and for $n$ sufficiently large, there exists an
$(n,M_1,M_2,J)$-code such that 
\begin{equation} \label{eq:kayEll}
  \frac{1}{n}\log (K_iL_i) \geq R_i-\delta
\end{equation}
for $i\in\{1,2\}$, and $P_e^{(n)}(r_1,r_2)\leq \epsilon$. In (\ref{eq:kayEll}),
we use $K_iL_i$ instead of $M_i$ since only $K_iL_i$ elements of
$[M_i]$ are used in calculating $P_e^{(n)}(r_1,r_2)$.
We define the $(r_1,r_2)$-error capacity region as the closure of the
set of all rates that are $(r_1,r_2)$-error achievable.

\section{Results} \label{sec:results}

We describe our results in this section. 
In Subsection \ref{subsec:coopRel}, we quantify the relation between 
cooperation under the CF model and reliability. 
In Subsection \ref{subsec:AvgMax}, we determine
the cooperation rate sufficient to guarantee equality between the 
maximal- and average-error capacity regions under the CF model. 
In Subsection \ref{subsec:negligible}, we define and study negligible
cooperation. Finally, we determine the reliability benefit of the conferencing
model in Subsection \ref{subsec:confRel}.  

\subsection{Cooperation and Reliability} \label{subsec:coopRel}

Our first result, Theorem \ref{thm:main}, says that if a rate pair is 
achievable for a MAC with a CF under average error, then sufficiently
increasing the capacities of the CF links ensures that the same rate 
pair is also achievable under a stricter notion of error. This result applies
to any memoryless MAC whose average-error capacity region is bounded. Prior
to stating this result, we introduce notation used in 
Theorem \ref{thm:main}. 

Define the nonnegative numbers $R_1^*$ and $R_2^*$ as 
the maximum of $R_1$ and $R_2$ over the average-error capacity region of a MAC with a 
$(\mathbf{C}_\mathrm{in},\mathbf{C}_\mathrm{out})$-CF and $J$ cooperation rounds.
Each rate is maximized when the other rate is set to zero. When one encoder transmits at rate zero,
cooperation through a CF is no more powerful then direct conferencing. Thus
$R_1^*$ and $R_2^*$ equal the corresponding maximal rates in the capacity region of the MAC 
with conferencing encoders \cite{WillemsMAC}. Hence,
\begin{align*}
  R_1^* &= \max_{X_1-U-X_2}\min\big\{I(X_1;Y|U,X_2)+C_{12},I(X_1,X_2;Y)\big\}\\
  R_2^* &= \max_{X_1-U-X_2}\min\big\{I(X_2;Y|U,X_1)+C_{21},I(X_1,X_2;Y)\big\},
\end{align*}
where 
$C_{12} = \min\{C_\mathrm{in}^1,C_\mathrm{out}^2\}$ and
$C_{21} = \min\{C_\mathrm{in}^2,C_\mathrm{out}^1\}$.
Note that since using multiple
conferencing rounds does not enlarge the average-error capacity region
for the 2-user MAC \cite{WillemsMAC},
$R_1^*$ and $R_2^*$ do not depend on $J$. 
 
\begin{thm}[Reliability under CF model] \label{thm:main} 
If $\tilde{J}\geq J+1$, and for $i\in\{1,2\}$,
\begin{align*}
  \tilde{C}^i_\mathrm{in} &> \min\{C^i_\mathrm{in}+\tilde{r}_i,R_i^*\} \\
  \tilde{C}^i_\mathrm{out} &> C^i_\mathrm{out},
\end{align*}
then the $(\tilde{r}_1,\tilde{r}_2)$-error capacity region 
of a MAC with a 
$(\mathbf{\tilde{C}}_\mathrm{in},\mathbf{\tilde{C}}_\mathrm{out})$-CF 
and $\tilde{J}$ rounds of cooperation 
contains the average-error capacity region of the same MAC with a 
$(\mathbf{C}_\mathrm{in},\mathbf{C}_\mathrm{out})$-CF and $J$ rounds of 
cooperation. 
Furthermore, if for $i\in\{1,2\}$, $\tilde{C}^i_\mathrm{in} > R_i^*$, 
$\tilde{J}= 1$ suffices. Similarly, $\tilde{J}=1$ suffices
when $\mathbf{C}_\mathrm{in}=\mathbf{0}$.
\end{thm}

A detailed proof of Theorem \ref{thm:main} appears in
Subsection \ref{subsec:main}. Roughly, the argument involves modifying 
an $(n,M_1,M_2,J)$ average-error 
code for a MAC with a $(\mathbf{C}_\mathrm{in},\mathbf{C}_\mathrm{out})$-CF to
get an $(n,\tilde{M}_1,\tilde{M}_2,\tilde{J})$ code for the same MAC
with a $(\mathbf{\tilde{C}}_\mathrm{in},\mathbf{\tilde{C}}_\mathrm{out})$-CF.
Our aim is to obtain small $(\tilde{r}_1,\tilde{r}_2)$ probability of
error and $\frac{1}{n}\log\tilde{M}_i$
only slightly smaller than $\frac{1}{n}\log M_i$ for $i\in\{1,2\}$.   
To achieve this goal, we first partition $\Lambda_n$, as given by (\ref{eq:Lambda}),
into $2^{n\tilde{r}_1}\times 2^{n\tilde{r}_2}$ blocks. 
We next construct a 
$2^{n\tilde{r}_1}\times 2^{n\tilde{r}_2}$ $(0,1)$-matrix, where entry 
$(k_1,k_2)$ equals zero if
the average of the $\lambda_n(m_1,m_2)$ entries in the corresponding 
block $(k_1,k_2)$ of $\Lambda_n$ is small, and equals one otherwise. 
(See Figure \ref{fig:matrix}.)

Next, we partition our (0,1)-matrix into blocks of size roughly $n\times n$. 
For each $i$, let $m_i$ denote the message of encoder $i$. In the first
cooperation round, encoder $i$ sends the first $n\tilde{r}_i$ bits of $m_i$ to the CF 
so that the CF knows the block in the (0,1)-matrix that contains $(m_1,m_2)$. 
If there is at least one zero entry in that
block, the CF sends the location of that entry back to each encoder using  
$\log n$ bits. Then encoder $i$
modifies the first $n\tilde{r}_i$ bits of its message and communicates with the CF 
over $J$ rounds using the original average-error code. As a result of transmitting
$(m_1,m_2)$ pairs that correspond to zeros in our $(0,1)$-matrix, the encoders
ensure a small $(r_1,r_2)$-probability of error. 
\begin{figure} 
	\begin{center}
		\includegraphics[scale=0.3]{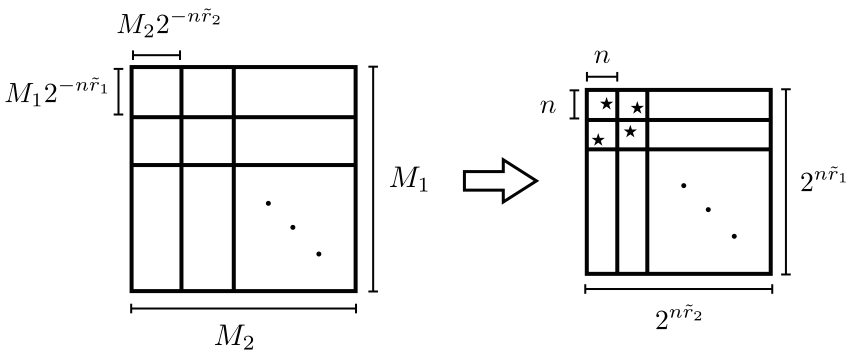}
		\caption{Left: The $M_1\times M_2$ matrix $\Lambda_n$ with entries 
		$\lambda_n(m_1,m_2)$. Right: The  
		$(0,1)$-matrix constructed from $\Lambda_n$.
		The stars indicate the location of the
		zeros.} \label{fig:matrix}
	\end{center}
\end{figure} 

It may be the case that not every block contains a zero entry. 
Lemma \ref{lem:matrix}, below, shows that if there is a sufficiently large
number of zeros in the (0,1)-matrix, then 
there exists a permutation of the rows and a permutation of the columns
such that each block of the permuted matrix contains at least one zero entry.
Since the original code has a small average error, it follows that
our (0,1)-matrix has a large number of zeros. The proof of Lemma \ref{lem:matrix}
appears in Subsection \ref{subsec:matrix}. 

\begin{lem} \label{lem:matrix}
Let $A=(a_{ij})_{i,j=1}^{m,n}$ be a $(0,1)$-matrix and let $N_A$
denote the number of ones in $A$.
Suppose $k$ is a positive integer smaller than or equal to $\min\{m,n\}$. For
any pair of permutations $(\pi_1,\pi_2)$, where $\pi_1$ is a permutation
on $[m]$ and $\pi_2$ is a permutation on $[n]$, and every
$(s,t)\in [\frac{m}{k}]\times[\frac{n}{k}]$, define the $k\times k$
matrix $B_{st}(\pi_1,\pi_2)$ as
\begin{equation*}
  B_{st}(\pi_1,\pi_2)=\big(a_{\pi_1(i)\pi_2(j)}\big),
\end{equation*}
where $i\in \{(s-1)k+1,\dots,sk\}$ and $j\in \{(t-1)k+1,\dots,tk\}$.
If 
\begin{equation*}
  \frac{mn}{k^2}\Big(\frac{N_Ae^2}{mn}\Big)^k<1,
\end{equation*}
then there exists a pair of permutations $(\pi_1,\pi_2)$ such that for
every $(s,t)$ the submatrix $B_{st}(\pi_1,\pi_2)$ contains at least one zero entry.
\end{lem}

\subsection{The Average- and Maximal-Error Capacity Regions} \label{subsec:AvgMax}

For every $(\mathbf{C}_\mathrm{in},\mathbf{C}_\mathrm{out})\in\mathbb{R}_{\geq 0}^4$,
let $\mathscr{C}_\mathrm{avg}^J(\mathbf{C}_\mathrm{in},\mathbf{C}_\mathrm{out})$
denote the average-error capacity region of a MAC with a 
$(\mathbf{C}_\mathrm{in},\mathbf{C}_\mathrm{out})$-CF with $J$ cooperation rounds.
Let  
\begin{equation*}
  \mathscr{C}_\mathrm{avg}(\mathbf{C}_\mathrm{in},\mathbf{C}_\mathrm{out})=
	\overline{
  \bigcup_{J=1}^\infty \mathscr{C}_\mathrm{avg}^J(\mathbf{C}_\mathrm{in},\mathbf{C}_\mathrm{out})},
\end{equation*}
where for any set $A\subseteq\mathbb{R}^2_{\geq 0}$, 
$\bar{A}$ denotes the closure of $A$. 
Define $\mathscr{C}_\mathrm{max}^J$ and $\mathscr{C}_\mathrm{max}$
similarly. 

We next introduce a generalization of the notion of sum-capacity
which is useful for the results of this section. Let $\mathscr{C}$ be 
a compact subset of $\mathbb{R}^2_{\geq 0}$. For every $\alpha\in [0,1]$
define
\begin{equation} \label{eq:Calpha}
  C^\alpha(\mathscr{C})
	=\max_{(x,y)\in\mathscr{C}}
	\big(\alpha x+(1-\alpha)y\big).
\end{equation}
Note that $C^\alpha$ is the value of the support function of $\mathscr{C}$ computed
with respect to the vector $(\alpha,1-\alpha)$ \cite[p. 37]{Schneider}. 
When $\mathscr{C}$ is the capacity region of a network,
$C^{1/2}(\mathscr{C})$ equals half the corresponding sum-capacity.

Now consider a MAC with a $(\mathbf{C}_\mathrm{in},\mathbf{C}_\mathrm{out})$-CF.
For every $\alpha\in [0,1]$, define 
\begin{equation*}
  C^\alpha_\mathrm{avg}(\mathbf{C}_\mathrm{in},\mathbf{C}_\mathrm{out})
	=C^\alpha\big(\mathscr{C}_\mathrm{avg}(\mathbf{C}_\mathrm{in},\mathbf{C}_\mathrm{out})\big).
\end{equation*}
Define $C^\alpha_\mathrm{max}(\mathbf{C}_\mathrm{in},\mathbf{C}_\mathrm{out})$
similarly.

Theorem \ref{thm:MaxEqAvg}, our main result, follows. 
This theorem states that
cooperation through a CF that has access to both messages results in 
a network whose maximal- and average-error capacity regions are identical.
We address the necessity of the assumption that the CF has access to both messages
in Proposition \ref{prop:MaxNeqAvg} at the end of this subsection. 

\begin{thm} \label{thm:MaxEqAvg}
For a given MAC $(\mathcal{X}_1\times\mathcal{X}_2,p(y|x_1,x_2),\mathcal{Y})$, let 
$\mathbf{C}_\mathrm{in}^*=(C^{*1}_\mathrm{in},C^{*2}_\mathrm{in})$ be any 
rate vector that satisfies
\begin{equation*}
 \min\{C_\mathrm{in}^{*1},C_\mathrm{in}^{*2}\}
 > \max_{p(x_1,x_2)}I(X_1,X_2;Y).
\end{equation*}
Then for every $\mathbf{C}_\mathrm{out}\in\mathbb{R}_{>0}^2$,
\begin{equation*}
  \mathscr{C}_\mathrm{max}(\mathbf{C}_\mathrm{in}^*,\mathbf{C}_\mathrm{out})
	=\mathscr{C}_\mathrm{avg}(\mathbf{C}_\mathrm{in}^*,\mathbf{C}_\mathrm{out}).
\end{equation*}
\end{thm} 

The following discussion gives the intuition behind the proof
of Theorem \ref{thm:MaxEqAvg}. Details follow in Subsection \ref{subsec:MaxEqAvg}. 
First, using
Theorem \ref{thm:main}, we show that for every 
$\mathbf{C}_\mathrm{out}=(C_\mathrm{out}^1,C_\mathrm{out}^2)$
and $\mathbf{\tilde{C}}_\mathrm{out}=(\tilde{C}_\mathrm{out}^1,\tilde{C}_\mathrm{out}^2)$
in $\mathbb{R}_{>0}^2$ with
$\tilde{C}^1_\mathrm{out}>C^1_\mathrm{out}$
and $\tilde{C}^2_\mathrm{out}>C^2_\mathrm{out}$, we have
\begin{equation*}
  \mathscr{C}_\mathrm{avg}(\mathbf{C}_\mathrm{in}^*,\mathbf{C}_\mathrm{out})
  \subseteq 
  \mathscr{C}_\mathrm{max}(\mathbf{C}_\mathrm{in}^*,\mathbf{\tilde{C}}_\mathrm{out}).
\end{equation*}
Note that $\mathscr{C}_\mathrm{avg}(\mathbf{C}_\mathrm{in}^*, \mathbf{C}_\mathrm{out})$
contains $\mathscr{C}_\mathrm{max}(\mathbf{C}_\mathrm{in}^*, \mathbf{C}_\mathrm{out})$. Thus
a continuity argument may be helpful in proving equality between 
the average- and maximal-error capacity regions. Since studying 
$C^\alpha$ is simpler than studying the
capacity region directly, we formulate our problem in terms of $C^\alpha$.
For every $\alpha\in [0,1]$, we have
\begin{equation} \label{eq:CalphaULB}
  C^\alpha_\mathrm{max}(\mathbf{C}_\mathrm{in}^*, \mathbf{C}_\mathrm{out})
	\leq C^\alpha_\mathrm{avg}(\mathbf{C}_\mathrm{in}^*, \mathbf{C}_\mathrm{out})
	\leq C^\alpha_\mathrm{max}(\mathbf{C}_\mathrm{in}^*, \mathbf{\tilde{C}}_\mathrm{out}).
\end{equation}
The next lemma, for fixed $\alpha\in [0,1]$, investigates the continuity of the mapping 
\begin{equation*}
  (\mathbf{C}_\mathrm{in},\mathbf{C}_\mathrm{out})\mapsto 
  C^\alpha(\mathbf{C}_\mathrm{in},\mathbf{C}_\mathrm{out}). 
\end{equation*}
In this lemma, $C^\alpha$ may be 
calculated with respect to either maximal- or average-error.
The proof is given in Subsection \ref{subsec:concavity}.
\begin{lem} \label{lem:continuity}
For every $\alpha\in [0,1]$, the mapping 
$C^\alpha(\mathbf{C}_\mathrm{in},\mathbf{C}_\mathrm{out})$ is
concave on $\mathbb{R}^4_{\geq 0}$ 
and thus continuous on $\mathbb{R}^4_{>0}$. 
\end{lem}
By combining the above lemma with (\ref{eq:CalphaULB}), 
it follows that for every
$\alpha\in [0,1]$ and $\mathbf{C}_\mathrm{out}\in \mathbb{R}^2_{>0}$, 
\begin{equation*}
  C^\alpha_\mathrm{max}(\mathbf{C}_\mathrm{in}^*,\mathbf{C}_\mathrm{out})
	=C^\alpha_\mathrm{avg}(\mathbf{C}_\mathrm{in}^*,\mathbf{C}_\mathrm{out}).
\end{equation*}
Since for a given capacity region $\mathscr{C}$, the mapping
$\alpha\mapsto C_\alpha(\mathscr{C})$ characterizes 
$\mathscr{C}$ precisely (see next lemma), for 
every $\mathbf{C}_\mathrm{out}\in \mathbb{R}^2_{>0}$, we have
\begin{equation*}
  \mathscr{C}_\mathrm{max}(\mathbf{C}_\mathrm{in}^*,\mathbf{C}_\mathrm{out})
	=\mathscr{C}_\mathrm{avg}(\mathbf{C}_\mathrm{in}^*,\mathbf{C}_\mathrm{out}).
\end{equation*}
\begin{lem} \label{lem:char}
Let $\mathscr{C}\subseteq\mathbb{R}^2_{\geq 0}$ be non-empty, compact, convex, 
and closed under projections onto the axes, that is, if $(x,y)$ is in $\mathscr{C}$,
then so are $(x,0)$ and $(0,y)$. Then
\begin{equation*}
  \mathscr{C}=\Big\{(x,y)\in\mathbb{R}^2_{\geq 0}\Big|
	\forall \alpha\in [0,1]:\alpha x+(1-\alpha)y\leq C^\alpha\Big\}.
\end{equation*}
\end{lem}
This result is well known and continues to hold for subsets of $\mathbb{R}^k_{\geq 0}$ for
any positive integer $k$. For completeness, we state and prove the general result in 
Subsection \ref{subsec:char}.

One question that arises from Theorem \ref{thm:MaxEqAvg} is whether it is necessary for the 
CF to have access to both messages to obtain identical maximal- and average-error capacity
regions. The next proposition shows that the mentioned condition is in fact required, 
that is, if the CF only has partial access to the messages, regardless of the capacities
of the CF output links, the average- and maximal-error regions sometimes differ. 
The proof is given in Subsection \ref{subsec:MaxNeqAvg}. 
\begin{prop} \label{prop:MaxNeqAvg}
There exists a MAC 
$(\mathcal{X}_1\times\mathcal{X}_2,p(y|x_1,x_2),\mathcal{Y})$ 
and $\mathbf{C}_\mathrm{in}\in\mathbb{R}^2_{>0}$
such that for every $\mathbf{C}_\mathrm{out}\in\mathbb{R}^2_{\geq 0}$,
\begin{equation*} 
  \mathscr{C}_\mathrm{max}(\mathbf{C}_\mathrm{in},\mathbf{C}_\mathrm{out})
  \neq
  \mathscr{C}_\mathrm{avg}(\mathbf{C}_\mathrm{in},\mathbf{C}_\mathrm{out}).
\end{equation*}
\end{prop}

\subsection{Negligible Cooperation} \label{subsec:negligible}

We begin by giving a rough description of the capacity region of 
a network containing edges of negligible capacity (Section \ref{sec:intro}). 
Let $\mathcal{N}$ be a memoryless network containing at 
most a single edge of negligible capacity and possibly other edges of positive capacity. 
For every $\delta>0$, let $\mathcal{N}(\delta)$ be the same network with the difference
that the edge with negligible capacity is replaced with an edge of capacity $\delta$. 
(See Figure \ref{fig:negligible}.) Then we say a rate vector 
is achievable over $\mathcal{N}$ if and only if for all $\delta>0$, that 
rate vector is achievable over $\mathcal{N}(\delta)$. Formally, if we denote
the capacity regions of $\mathcal{N}$ and $\mathcal{N}(\delta)$ with 
$\mathscr{C}(\mathcal{N})$ and $\mathscr{C}(\mathcal{N}(\delta))$, respectively, then 
\begin{equation*}
\mathscr{C}(\mathcal{N})=\bigcap_{\delta>0}
\mathscr{C}(\mathcal{N}(\delta)).
\end{equation*}
We define achievability over networks with 
multiple edges of negligible capacity inductively.  
\begin{figure} 
	\begin{center}
		\includegraphics[scale=0.3]{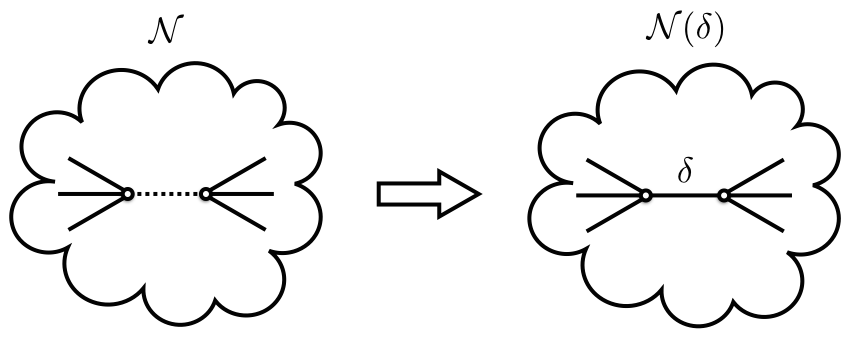}
		\caption{Left: A network $\mathcal{N}$ with a single edge of ``negligible 
		capacity.'' Right: The network $\mathcal{N}(\delta)$, where the negligible
		capacity edge of $\mathcal{N}$ is replaced with an edge of capacity $\delta>0$.} 
		\label{fig:negligible}
	\end{center}
\end{figure} 
 
Based on the above discussion, we define the capacity region 
of a MAC with a CF that has complete access to both messages 
and output edges of negligible capacity as
\begin{equation*}
  \bigcap_{\mathbf{C}_\mathrm{out}\in\mathbb{R}^2_{>0}}
	\mathscr{C}(\mathbf{C}_\mathrm{in}^*,\mathbf{C}_\mathrm{out}),
\end{equation*}
where in the intersection we place either $\mathscr{C}_\mathrm{avg}$ or 
$\mathscr{C}_\mathrm{max}$ depending on whether the average- or maximal-error
capacity region is desired.
From Theorem \ref{thm:MaxEqAvg} it follows that for every MAC,
\begin{equation} \label{eq:negEdgeMaxAvg}
  \bigcap_{\mathbf{C}_\mathrm{out}\in\mathbb{R}_{>0}^2}
	\mathscr{C}_\mathrm{max}(\mathbf{C}_\mathrm{in}^*,\mathbf{C}_\mathrm{out})
	=\bigcap_{\mathbf{C}_\mathrm{out}\in\mathbb{R}_{>0}^2}
	\mathscr{C}_\mathrm{avg}(\mathbf{C}_\mathrm{in}^*,\mathbf{C}_\mathrm{out})
	\supseteq \mathscr{C}_\mathrm{avg}(\mathbf{0},\mathbf{0}),
\end{equation}
where $\mathbf{0}=(0,0)$. Thus even negligible cooperation increases
reliability. That is, a negligible cooperation rate suffices to guarantee 
a small maximal probability of error for rate pairs that without cooperation,
may only be achievable with small average probability of error. 

The reliability gain of negligible cooperation is closely related to the question
of the continuity of the capacity region of a network with respect to its edges.
Using the ideas discussed above, Proposition \ref{prop:discontinuity} provides conditions 
under which  $C^\alpha_\mathrm{max}(\mathbf{C}_\mathrm{in},\mathbf{C}_\mathrm{out})$ is not 
continuous with respect to $\mathbf{C}_\mathrm{out}$. The proof is given in 
Subsection \ref{subsec:discontinuity}.
\begin{prop} \label{prop:discontinuity}
Fix $\alpha\in (0,1)$ and $\mathbf{C}_\mathrm{in}\in\mathbb{R}^2_{>0}$. 
Given any MAC for which
\begin{equation} \label{eq:CalphaZeroOne}
C_\mathrm{avg}^\alpha(\mathbf{0},\mathbf{0})>C_\mathrm{max}^\alpha(\mathbf{0},\mathbf{0}),
\end{equation}
$C^\alpha_\mathrm{max}(\mathbf{C}_\mathrm{in},\mathbf{C}_\mathrm{out})$ is not 
continuous with respect to $\mathbf{C}_\mathrm{out}$ at $\mathbf{C}_\mathrm{out}=\mathbf{0}$.
\end{prop}

In Subsection \ref{subsec:contractionMAC}, we show that Dueck's contraction MAC \cite{Dueck} 
is an example of a MAC that satisfies (\ref{eq:CalphaZeroOne}) for 
\emph{every} $\alpha\in (0,1)$. This results in the next corollary. 
\begin{cor} \label{cor:discontinuity}
There exists a MAC where for all $\mathbf{C}_\mathrm{in}\in\mathbb{R}_{>0}^2$ and 
$\alpha\in (0,1)$, 
$C^\alpha_\mathrm{max}(\mathbf{C}_\mathrm{in},\mathbf{C}_\mathrm{out})$
is not continuous with respect to $\mathbf{C}_\mathrm{out}$ at 
$\mathbf{C}_\mathrm{out}=\mathbf{0}$. 
\end{cor}

For the average-error capacity region of the MAC, less is known. For some MACs and all
$\mathbf{C}_\mathrm{in}\in\mathbb{R}^2_{>0}$, the directional derivative of 
$C^{1/2}_\mathrm{avg}(\mathbf{C}_\mathrm{in},\mathbf{C}_\mathrm{out})$ 
at $\mathbf{C}_\mathrm{out}=\mathbf{0}$ equals infinity for all unit vectors in $\mathbb{R}_{>0}^2$
\cite{kUserMACLong}. The question of whether 
$C^{1/2}_\mathrm{avg}(\mathbf{C}_\mathrm{in},\mathbf{C}_\mathrm{out})$, for a fixed
$C_\mathrm{in}\in\mathbb{R}_{>0}^2$, is continuous with respect to $\mathbf{C}_\mathrm{out}$
at $\mathbf{C}_\mathrm{out}=\mathbf{0}$ for such MACs remains open. 

\subsection{The Conferencing Encoders Model} \label{subsec:confRel}

In this subsection, we study the reliability benefit of cooperation under
the conferencing encoders model \cite{WillemsMAC}, 
in addition to the maximal- and average-error capacity regions of the
MAC with conferencing. 

Theorem \ref{thm:main} quantifies the reliability benefit of cooperation via a CF.
The next proposition does the same for cooperation via conferencing \cite{WillemsMAC}. 
The proof is given Subsection \ref{subsec:confReliability}.
\begin{prop}[Reliability under conferencing] \label{prop:confReliability}
Fix $(C_{12},C_{21})\in\mathbb{R}^2_{\geq 0}$. Then for any MAC with 
$(C_{12},C_{21})$-conferencing, the average- and $(C_{12},C_{21})$-error 
capacity regions are identical. 
\end{prop}

Let $(C_{12},C_{21})\in\mathbb{R}^2_{\geq 0}$
and $\mathscr{C}_\mathrm{conf}(C_{12},C_{21})$ denote the maximal- or
average-error capacity region of the MAC with $(C_{12},C_{21})$-conferencing. Then
for every $(C_{12},C_{21})\in\mathbb{R}^2_{\geq 0}$, 
\begin{equation} \label{eq:confCF}
  \mathscr{C}_\mathrm{conf}(C_{12},C_{21})
	=\mathscr{C}(\mathbf{C}_\mathrm{in},\mathbf{C}_\mathrm{out}),
\end{equation} 
where 
\begin{align}
  \mathbf{C}_\mathrm{in} &= (C_{12},C_{21}) \label{eq:CinConf}\\
  \mathbf{C}_\mathrm{out} &= (C_{21},C_{12}). \label{eq:CoutConf}
\end{align}
Equation (\ref{eq:confCF}) follows from the fact that for a CF for which the 
output link capacity to each encoder is at least as large as the input link
capacity from the other encoder, the strategy where the CF forwards its received 
information from one encoder to the other is optimal. 
Combining Proposition \ref{prop:MaxNeqAvg} with (\ref{eq:confCF})
implies that direct cooperation via conferencing does not necessarily lead 
to identical maximal- and average-error capacity regions.  This is stated 
formally in the next corollary.
\begin{cor}
There exists a MAC and $(C_{12},C_{21})\in\mathbb{R}^2_{>0}$ such that
\begin{equation*}
  \mathscr{C}_\mathrm{conf,max}(C_{12},C_{21})\neq
  \mathscr{C}_\mathrm{conf,avg}(C_{12},C_{21}).
\end{equation*}
\end{cor}

We next study the continuity of the conferencing capacity region 
with respect to the capacities of the conferencing links. 
For every $\alpha\in [0,1]$, define
\begin{equation} \label{eq:CF2Conf}
  C^\alpha_\mathrm{conf}(C_{12},C_{21})
	=C^\alpha(\mathbf{C}_\mathrm{in},\mathbf{C}_\mathrm{out}),
\end{equation}
where
$\mathbf{C}_\mathrm{in}$ and $\mathbf{C}_\mathrm{out}$ are given by
(\ref{eq:CinConf}) and (\ref{eq:CoutConf}).
Our next result considers the continuity of $C^\alpha_\mathrm{conf}$
for various values of $\alpha\in [0,1]$. 
\begin{prop} \label{prop:conf}
For every MAC $(\mathcal{X}_1\times\mathcal{X}_2,p(y|x_1,x_2),\mathcal{Y})$, 
the following statements are true. 

(a) For every $\alpha\in [0,1]$, $C^\alpha_\mathrm{conf,avg}$
is continuous on $\mathbb{R}^2_{\geq 0}$.

(b) For every $\alpha\in [0,1]$, $C^\alpha_\mathrm{conf,max}$
is continuous on $\mathbb{R}^2_{>0}$, and for $\alpha=1/2$, 
$C^{1/2}_\mathrm{conf,max}$ is continuous at the point $(0,0)$.
\end{prop}

\section{Proofs}

\subsection{Theorem \ref{thm:main} (Reliability under the CF model)} \label{subsec:main}

Our aim is to show that if $\tilde{J}\geq J+1$, and 
\begin{align*}
  \tilde{C}_\mathrm{in}^i &> \min\{C_\mathrm{in}^i+\tilde{r}_i,R_i^*\}\\
	\tilde{C}_\mathrm{out}^i &> C_\mathrm{out}^i
\end{align*}
for $i\in\{1,2\}$, then the $(\tilde{r}_1,\tilde{r}_2)$-error capacity region
of the MAC with a $(\mathbf{\tilde{C}}_\mathrm{in},\mathbf{\tilde{C}}_\mathrm{out})$-CF
and $\tilde{J}$ cooperation rounds
contains the average-error capacity region of the same MAC with a
$(\mathbf{C}_\mathrm{in},\mathbf{C}_\mathrm{out})$-CF and $J$ cooperation rounds. In 
addition, here we show that if for $i\in\{1,2\}$, 
$\tilde{C}_\mathrm{in}^i>R_i^*$, $\tilde{J}=1$ suffices. Similarly, $\tilde{J}=1$
suffices when $\mathbf{C}_\mathrm{in}=\mathbf{0}$. Also recall $R_1^*$ and 
$R_2^*$ are defined as the maximum of $R_1$ and $R_2$ over the capacity region of 
a MAC with a $(\mathbf{C}_\mathrm{in},\mathbf{C}_\mathrm{out})$-CF and $J$ cooperation 
rounds. Our proof follows \cite{WillemsBC}, where Willems proves that the maximal-
and average-error capacity regions of the broadcast channel are identical. 

Suppose $(R_1,R_2)$ is in the average-error capacity region of the
MAC with a $(\mathbf{C}_\mathrm{in},\mathbf{C}_\mathrm{out})$-CF and $J$-round
cooperation. Assume $\tilde{r}_1,\tilde{r}_2,R_1,R_2$ 
are all positive. We discuss the case where some of these quantities are zero
at the end of this subsection.
Fix $\epsilon,\delta>0$. Then for sufficiently large $N$
and any $n>N$, there exists an $(n,M_1,M_2,J)$-code such that for $i=1,2$,
\begin{align}
  \log \big|\mathcal{U}_i^J\big|&\leq nC_\mathrm{in}^i \label{eq:UijBound}\\
  \log \big|\mathcal{V}_i^J\big|&\leq nC_\mathrm{out}^i \label{eq:VijBound}\\
  \frac{1}{n}\log M_i &\geq R_i-\delta \label{eq:MiBound}
\end{align}
and $P_{e,\mathrm{avg}}^{(n)}\leq \epsilon$. In addition, from Fano's inequality
it follows that for sufficiently large $n$,
\begin{equation} \label{eq:Fano}
  \frac{1}{n}\log M_i\leq R_i^*+\delta.
\end{equation}
Let $K_*=\lceil n(R_1^*+R_2^*+2\delta)\rceil$.
For $i\in \{1,2\}$, define $K_i=\min\{K_*\lfloor 2^{n\tilde{r}_i}\rfloor,M_i\}$
and $L_i=\lfloor M_i/K_i\rfloor$. 
From the set $[M_1]$ choose the $K_1L_1$ messages that have the smallest
\begin{equation*}
  \sum_{m_2=1}^{M_2}\lambda_n(m_1,m_2),
\end{equation*}
and renumber them as $\{1,\dots,K_1L_1\}$. 
Similarly, from the set $[M_2]$ choose $K_2L_2$ 
messages that have the smallest 
\begin{equation*}
  \sum_{m_1=1}^{K_1L_1}
	\lambda_n(m_1,m_2)
\end{equation*}
and renumber them as $\{1,\dots,K_2L_2\}$. Then
\begin{align} \label{eq:modPe}
  \MoveEqLeft
  \frac{1}{K_1L_1K_2L_2}
	\sum_{m_1=1}^{K_1L_1}\sum_{m_2=1}^{K_2L_2}
  \lambda_n(m_1,m_2)\notag\\
  &\leq \frac{1}{K_1L_1M_2}
	\sum_{m_1=1}^{K_1L_1}\sum_{m_2=1}^{M_2}
  \lambda_n(m_1,m_2) \notag\\
  &\leq \frac{1}{M_1M_2}\sum_{m_1=1}^{M_1}\sum_{m_2=1}^{M_2}
  \lambda_n(m_1,m_2)\leq \epsilon
\end{align}
Next, for every $(k_1,k_2)\in [K_1]\times [K_2]$, define $a_{k_1k_2}$ as
\begin{equation*}
a_{k_1k_2}=
\begin{cases}
1 &\text{if }\sum_{S_{1,k_1}\times S_{2,k_2}}
\lambda_n(m_1,m_2)>L_1L_2e^3\epsilon\\
0 &\text{otherwise,}
\end{cases}
\end{equation*}
where $S_{1,k_1}$ and $S_{2,k_2}$ are defined by (\ref{eq:Sk}).
Let $N_A$ denote the number of ones in the $K_1\times K_2$ matrix 
$A=(a_{k_1k_2})_{k_1,k_2}$. Then
\begin{align} \label{eq:NA}
  N_A &= \sum_{k_1,k_2} a_{k_1k_2} \notag\\
  &\leq \frac{1}{L_1L_2e^3\epsilon} \sum_{k_1,k_2}
  \sum_{S_{1,k_1}\times S_{2,k_2}}\lambda_n(m_1,m_2)\notag\\
  &= \frac{1}{L_1L_2e^3\epsilon}
  \sum_{m_1=1}^{K_1L_1}
	\sum_{m_2=1}^{K_2L_2}\lambda_n(m_1,m_2)\notag\\
  &\leq K_1K_2e^{-3},
\end{align}
where the last inequality follows from (\ref{eq:modPe}).

Next define $\alpha$ as 
\begin{equation*}
  \alpha=\frac{K_1K_2}{K_*^2}
  \Big(\frac{N_A e^2}{K_1K_2}\Big)^{K_*}.
\end{equation*}
Note that $\alpha$ can be bounded from above by
\begin{align*}
  \alpha &\overset{(a)}{\leq} \frac{K_1K_2}{K_*^2e^{K_*}}\\
  &\overset{(b)}{\leq} 2^{n(R_1^*+R_2^*+2\delta)-K_*\log e-2\log K_*}
  \overset{(c)}{<}1,
\end{align*}
where $(a)$ follows from (\ref{eq:NA}),
$(b)$ follows from (\ref{eq:Fano}) and the fact that
$K_i\leq M_i$, and $(c)$ follows from the fact that 
$K_*=\lceil n(R_1+R_2+2\delta)\rceil$. Thus by Lemma \ref{lem:matrix},
there exist permutations $\pi_1$ and $\pi_2$ on the sets 
$[K_1]$ and $[K_2]$, respectively, such that if we partition the 
matrix $(a_{\pi_1(k_1)\pi_2(k_2)})$ into 
blocks of size $K_*\times K_*$, then there is at least one zero in each block. 
For $i\in\{1,2\}$, define 
\begin{equation*}
  K_i^*=\lfloor K_i/K_*\rfloor.
\end{equation*}
Note that the partition of the matrix $(a_{\pi_1(k_1)\pi_2(k_2)})$
contains at least $K_1^*\times K_2^*$ blocks. 

Next we use the partition defined above to construct
a coding strategy that achieves a rate pair sufficiently 
close to $(R_1,R_2)$ under $(\tilde{r}_1,\tilde{r}_2)$-error. 
For $i\in\{1,2\}$, encoder $i$ splits its message as $m_i=(k_i,\ell_i)\in [K_i]\times [L_i]$ 
and sends $k_i$ to the CF. Let $(\pi_1(k_1^*),\pi_2(k_2^*))$ be the good entry in the 
$K_*\times K_*$ block containing the pair $(\pi_1(k_1),\pi_2(k_2))$. 
For $i\in \{1,2\}$, the CF sends the difference
$\pi_i(k_i^*)-\pi_i(k_i)$ (mod $K_*$) back to encoder $i$. 
Encoder 1 and encoder 2 then use the original average-error
code with $J$ rounds of cooperation to transmit the message pair $(m_1^*,m_2^*)$ where 
for $i\in \{1,2\}$, $m_i^*=(\pi_i(k_i^*),\ell_i)$. By combining
$(\ref{eq:UijBound})$, $(\ref{eq:VijBound})$, and the fact that 
$K_i\leq K_*2^{n\tilde{r}_i}$, we see that for sufficiently large $n$,
\begin{align*}
  \frac{1}{n}\log|\mathcal{U}_i^J|K_i
  &\leq C_\mathrm{in}^i+\tilde{r}_i +\frac{1}{n}\log(1+n(R_1^*+R_2^*+2\delta))
  <\tilde{C}_\mathrm{in}^i\\
  \frac{1}{n}\log|\mathcal{V}_i^J|K_* &\leq
  C_\mathrm{out}^i+\frac{1}{n}\log(1+n(R_1^*+R_2^*+2\delta))
  <\tilde{C}_\mathrm{out}^i.
\end{align*}
Thus the rate achieved by encoder $i$ under an $(\tilde{r}_1,\tilde{r}_2)$ 
notion of error is at least as large as
\begin{equation*}
  \frac{1}{n}\log K_i^*L_i
  = \frac{1}{n}\log\Big\lfloor\frac{K_i}{K_*}\Big\rfloor
  \Big\lfloor\frac{M_i}{K_i}\Big\rfloor.
\end{equation*}

We next find a lower bound for the above expression. If
$\tilde{r}_i<R_i$, then for sufficiently large $n$,
$K_i=K_*\lfloor 2^{n\tilde{r}_i}\rfloor$, and the above quantity is at least as large as
\begin{align*}
  \MoveEqLeft
  \frac{1}{n}\log\big(2^{n\tilde{r}_i}-1\big)
  \Big(\frac{1}{K_*}2^{n(R_i-\delta-\tilde{r}_i)}-1\Big)\\
  &\geq R_i-\delta+\frac{1}{n}\log\big(1-2^{-n\tilde{r}_i}\big)
  \Big(
  \frac{1}{n(R_1^*+R_2^*+2\delta)+1}-2^{-n(R_i-\delta-\tilde{r}_i)}\Big)\\
  &> R_i-2\delta.
\end{align*}
On the other hand, if $\tilde{r}_i\geq R_i$, then for sufficiently large 
$n$, $K_i\geq 2^{n(R_i-\delta)}$ for $i\in\{1,2\}$. Thus 
\begin{align*}
  \frac{1}{n}\log\Big\lfloor\frac{K_i}{K_*}\Big\rfloor\Big\lfloor\frac{M_i}{K_i}\Big\rfloor
  &\geq \frac{1}{n}\log\Big\lfloor\frac{K_i}{K_*}\Big\rfloor\\
  &\geq
  \frac{1}{n}\log\Big(\frac{2^{n(R_i-\delta)}}{1+n(R_1^*+R_2^*+2\delta)}-1\Big)\\
  &= R_i-\delta+\frac{1}{n}\log\Big(\frac{1}{1+n(R_1^*+R_2^*+2\delta)}
  -2^{-n(R_i-\delta)}\Big)\\
  &> R_i-2\delta,
\end{align*}

If for $i\in\{1,2\}$, $C_\mathrm{in}^i> R_i^*$, 
encoder $i$ can send $m_i$ directly to the CF. The CF computes $(m_1^*,m_2^*)$ and 
sends its corresponding output from the original average-error code, in addition to 
$\pi_i(k_i^*)-\pi_i(k_i)$ (mod $K_*$), back to encoder $i$. 
Thus a single round of cooperation suffices in this case. 

On the other hand, when $\mathbf{C}_\mathrm{in}=\mathbf{0}$, no cooperation is possible 
in the original average-error code. This means that in the new code, we only need 
the first cooperation round to guarantee small $(\tilde{r}_1,\tilde{r}_2)$-error. 
Thus it suffices to have $\tilde{J}=1$.

When either $\min\{\tilde{r}_1,\tilde{r}_2\}=0$ or $\min\{R_1,R_2\}=0$, we apply a
similar argument, but instead of using Lemma \ref{lem:matrix}, we use its
corresponding vector version, which we state below. 
\begin{lem}[Vector Version] \label{lem:vector}
Let $A=(a_i)_{i=1}^{m}$ be a $(0,1)$-vector and let $N_A$
denote the number of ones in $A$, that is, 
\begin{equation*}
  N_A=\sum_{i=1}^m a_{i}.
\end{equation*}
Suppose $k$ is a positive integer smaller or equal to $m$. For
any permutation $\pi$ on $[m]$ and
$s\in [\frac{m}{k}]$, let $B_s(\pi)$ denote the vector
\begin{equation*}
  B_s(\pi)=\big(a_{\pi(i)}\big)_{i=(s-1)k+1}^{sk},
\end{equation*}
If 
\begin{equation*}
  \frac{m}{k}\Big(\frac{N_Ae}{m}\Big)^k<1,
\end{equation*}
then there exists a permutation $\pi$ such that for
every $s\in [\frac{m}{k}]$, the vector $B_{s}(\pi)$ contains at least one zero.
\end{lem}

\subsection{Lemma \ref{lem:matrix} (Existence of good permutations)} \label{subsec:matrix}
Let $A=(a_{ij})_{i,j=1}^{m,n}$ be a $(0,1)$-matrix. We apply the probabilistic method. Let
$\Pi_1$ and $\Pi_2$ be independent and uniformly distributed random variables on 
the set of all permutations of $[m]$ and $[n]$, respectively. Let $N_A$
denote the number of ones in $A$, that is, 
\begin{equation*}
  N_A=\sum_{i=1}^m\sum_{j=1}^n a_{ij}.
\end{equation*}
For $(s,t)\in [\frac{m}{k}]\times[\frac{n}{k}]$, define the $k\times k$
matrix $B_{st}(\Pi_1,\Pi_2)$ as
\begin{equation*}
  B_{st}(\Pi_1,\Pi_2)=\big(a_{\Pi_1(i)\Pi_2(j)}\big),
\end{equation*}
where $i\in \{(s-1)k+1,\dots,sk\}$ and $j\in \{(t-1)k+1,\dots,tk\}$.
Let $J_k$ denote the $k\times k$ matrix consisting of all ones. 
By the union bound,
\begin{equation} \label{eq:unionbound}
  \pr\Big\{\exists (s,t):B_{st}(\Pi_1,\Pi_2)=J_k\Big\}
  \leq \frac{mn}{k^2}\pr\big\{B_{11}(\Pi_1,\Pi_2)=J_k\big\}.
\end{equation}
We next find an upper bound for $\pr\{B_{11}(\Pi_1,\Pi_2)=J_k\}$. Consider the 
pairs $(S_1,S_2)$ and $(\tau_1,\tau_2)$, where $S_1\subseteq [m]$, $S_2\subseteq [n]$,
$|S_1|=|S_2|=k$, and $\tau_1$ and $\tau_2$ are permutations on the set $[k]$. 
In addition, denote the elements of $S_1$ and $S_2$ with
\begin{align*}
  S_1 &= \{i_1,\dots,i_k\}\\
  S_2 &= \{j_1,\dots,j_k\}.
\end{align*}
Define  $E_{S_1S_2}^{\tau_1\tau_2}$ as the event where for all $\ell\in [k]$, 
$\Pi_1(\ell)=i_{\tau_1(\ell)}$ and $\Pi_2(\ell)=j_{\tau_2(\ell)}$. In other
words, when $E_{S_1S_2}^{\tau_1\tau_2}$ occurs, $B_{11}(\Pi_1,\Pi_2)$ is 
a (permuted) submatrix of $A$ with row indices $(i_{\tau_1(\ell)})_{\ell\in [k]}$
and column indices $(j_{\tau_2(\ell)})_{\ell\in [k]}$.
We have
\begin{align*}
  \pr\big\{B_{11}=J_k\big\} 
  &\leq \pr\Big\{\forall \ell\in [k]:a_{\Pi_1(\ell)\Pi_2(\ell)}=1\Big\}\\
  &= \sum_{S_1,S_2,\tau_1,\tau_2}
  \pr\big(E_{S_1S_2}^{\tau_1\tau_2}\big)\pr\Big\{\forall \ell\in [k]:a_{\Pi_1(\ell)\Pi_2(\ell)}
  =1\Big|E_{S_1S_2}^{\tau_1\tau_2}\Big\}.
\end{align*}
Note that
\begin{align*}
  \pr\big(E_{S_1S_2}^{\tau_1\tau_2}\big)
	&= \pr\Big\{\forall\ell\in [k]:\Pi_1(\ell)=i_{\tau_1(\ell)},
	\Pi_2(\ell)=j_{\tau_2(\ell)}\Big\}\\
	&\overset{(a)}{=} \pr\Big\{\forall\ell\in [k]:\Pi_1(\ell)=i_{\tau_1(\ell)}
	\Big\}\times\pr\Big\{\forall\ell\in [k]:
	\Pi_2(\ell)=j_{\tau_2(\ell)}\Big\}\\
	&\overset{(b)}{=} \frac{(m-k)!}{m!}\times\frac{(n-k)!}{n!}
	=\frac{1}{(k!)^2\binom{m}{k}\binom{n}{k}},
\end{align*}
where $(a)$ follows from the independence of $\Pi_1$ and $\Pi_2$, 
and $(b)$ follows from the fact that $\Pi_1$ and $\Pi_2$ are uniformly distributed. 
Furthermore, 
\begin{equation*}
  \pr\Big\{\forall \ell\in [k]:a_{\Pi_1(\ell)\Pi_2(\ell)}=1\Big|E_{S_1S_2}^{\tau_1\tau_2}\Big\}
	=\mathbf{1}\big\{\forall \ell\in [k]:a_{i_{\tau_1(\ell)} j_{\tau_2(\ell)}}=1\big\}.
\end{equation*}
Thus
\begin{equation*}
  \pr\big\{B_{11}=J_k\big\} 
  \leq \frac{1}{(k!)^2\binom{m}{k}\binom{n}{k}}
	\sum_{S_1,S_2}\sum_{\tau_1,\tau_2}
	\mathbf{1}\big\{\forall \ell\in [k]:a_{i_{\tau_1(\ell)} j_{\tau_2(\ell)}}=1\big\},
\end{equation*}
Note that for a fixed pair $(S_1,S_2)$, 
\begin{align*}
  \MoveEqLeft
  \sum_{\tau_1}\sum_{\tau_2}
	\mathbf{1}\big\{\forall \ell\in [k]:a_{i_{\tau_1(\ell)} j_{\tau_2(\ell)}}=1\big\}\\
	&= \sum_{\tau_1}\sum_{\tau_2}
	\mathbf{1}\big\{\forall \ell\in [k]:a_{i_{(\tau_1\circ\tau_2^{-1})(\ell)} j_{\ell}}=1\big\}\\
	&= k!\sum_{\tau}
	\mathbf{1}\big\{\forall \ell\in [k]:a_{i_{\tau(\ell)} j_{\ell}}=1\big\}
\end{align*}
which equals $k!$ times the number of $k$-subsets of $S_1\times S_2$ that consist only 
of ones and have exactly one entry in each row and each column. 
Summing over all $S_1$ and $S_2$, we see that that the total number of such subsets is
bounded from above by $\binom{N_A}{k}$. Thus
\begin{equation} \label{eq:B11}
  \pr\big\{B_{11}=J_k\big\} 
  \leq\frac{k!\binom{N_A}{k}}{(k!)^2\binom{m}{k}\binom{n}{k}}
  =\frac{\binom{N_A}{k}}{k!\binom{m}{k}\binom{n}{k}}.
\end{equation}
Therefore, 
\begin{align*}
  \pr\Big\{\exists (s,t):B_{st}(\Pi_1,\Pi_2)=J_k\Big\}
  &\overset{(a)}{\leq} \frac{mn}{k^2}\times\frac{\binom{N_A}{k}}{k!\binom{m}{k}\binom{n}{k}}\\
  &\overset{(b)}{\leq} \frac{mn}{k^2}\times\frac{\big(\frac{N_Ae}{k}\big)^k}
  {\big(\frac{m}{e}\big)^k\big(\frac{n}{k}\big)^k}\\
  &= \frac{mn}{k^2}\Big(\frac{N_Ae^2}{mn}\Big)^k,
\end{align*}
where $(a)$ follows from combining (\ref{eq:unionbound}) and (\ref{eq:B11}),
and $(b)$ follows from Lemma \ref{lem:binom} \cite[Appendix C.1]{CLRS} which is stated below.
\begin{lem} \label{lem:binom}
For integers $k$ and $n$ that satisfy $1\leq k\leq n$, we have
\begin{equation*}
  \big(\frac{n}{k}\big)^k \leq
  \frac{1}{k!}\big(\frac{n}{e}\big)^k\leq
  \binom{n}{k}\leq \big(\frac{ne}{k}\big)^k.
\end{equation*}
\end{lem}

\subsection{Theorem \ref{thm:MaxEqAvg} (Average- and maximal-error capacity regions under CF model)} \label{subsec:MaxEqAvg}

Let $\mathbf{C}_\mathrm{out}=(C_\mathrm{out}^1,C_\mathrm{out}^2)$
and $\mathbf{\tilde{C}}_\mathrm{out}=(\tilde{C}_\mathrm{out}^1,\tilde{C}_\mathrm{out}^2)$
be elements of $\mathbb{R}_{>0}^2$ such that for $i\in\{1,2\}$, 
$\tilde{C}^i_\mathrm{out}>C^i_\mathrm{out}$. In Theorem \ref{thm:main}, 
for $i\in\{1,2\}$, set $\tilde{C}_\mathrm{in}^i=C_\mathrm{in}^i=C_\mathrm{in}^{*i}>R_i^*$
and $\tilde{r}_i>R_i^*$. Then 
\begin{equation*}
  \mathscr{C}_\mathrm{avg}(\mathbf{C}_\mathrm{in}^*,\mathbf{C}_\mathrm{out})
  \subseteq \mathscr{C}_\mathrm{max}(\mathbf{C}_\mathrm{in}^*,\mathbf{\tilde{C}}_\mathrm{out}).	
\end{equation*}
Thus for every $\alpha\in [0,1]$,
\begin{equation*}
  C^\alpha_\mathrm{max}(\mathbf{C}_\mathrm{in}^*,\mathbf{C}_\mathrm{out}) 
  \leq C^\alpha_\mathrm{avg}(\mathbf{C}_\mathrm{in}^*,\mathbf{C}_\mathrm{out})
  \leq C^\alpha_\mathrm{max}(\mathbf{C}_\mathrm{in}^*,\mathbf{\tilde{C}}_\mathrm{out}).
\end{equation*}
Since by Theorem \ref{lem:concavity}, $C^\alpha_\mathrm{max}$
is continuous on $\mathbb{R}^4_{>0}$, taking the limits
$\tilde{C}^1_\mathrm{out}\rightarrow (C^1_\mathrm{out})^+$ and
$\tilde{C}^2_\mathrm{out}\rightarrow (C^2_\mathrm{out})^+$,
results in
\begin{equation*}
  C^\alpha_\mathrm{max}(\mathbf{C}_\mathrm{in}^*,\mathbf{C}_\mathrm{out})
	= C^\alpha_\mathrm{avg}(\mathbf{C}_\mathrm{in}^*,\mathbf{C}_\mathrm{out}).
\end{equation*}
Since this result holds for every $\alpha\in [0,1]$, by
Theorem \ref{lem:char} it follows
that for every $\mathbf{C}_\mathrm{out}\in\mathbb{R}^2_{>0}$,
\begin{equation*}
  \mathscr{C}_\mathrm{max}(\mathbf{C}_\mathrm{in}^*,\mathbf{C}_\mathrm{out})
	=\mathscr{C}_\mathrm{avg}(\mathbf{C}_\mathrm{in}^*,\mathbf{C}_\mathrm{out}).
\end{equation*}

\subsection{Lemma \ref{lem:continuity} 
(Concavity and continuity of \texorpdfstring{$C^\alpha$}{Calpha})} 
\label{subsec:concavity}

We first show that 
$\mathscr{C}(\mathbf{C}_\mathrm{in},\mathbf{C}_\mathrm{out})$ is concave
with respect to $(\mathbf{C}_\mathrm{in},\mathbf{C}_\mathrm{out})$
for both maximal- and average-error. This result
was originally stated for the average-error capacity region in \cite{Noorzad2}.

\begin{lem} \label{lem:concavity}
For every $(\mathbf{C}^a_\mathrm{in},\mathbf{C}^a_\mathrm{out})$ and
$(\mathbf{C}^b_\mathrm{in},\mathbf{C}^b_\mathrm{out})$ in $\mathbb{R}^4_{\geq 0}$
and $\mu\in [0,1]$, define $(\mathbf{C}^\mu_\mathrm{in},\mathbf{C}^\mu_\mathrm{out})$
as
\begin{equation*}
  (\mathbf{C}_\mathrm{in}^\mu,\mathbf{C}_\mathrm{out}^\mu)
	=\mu (\mathbf{C}_\mathrm{in}^a,\mathbf{C}_\mathrm{out}^b)
	+(1-\mu) (\mathbf{C}_\mathrm{in}^b,\mathbf{C}_\mathrm{out}^b).
\end{equation*}
Then
\begin{equation*}
  \mathscr{C}(\mathbf{C}^\mu_\mathrm{in},\mathbf{C}^\mu_\mathrm{out})
	\supseteq \mu\mathscr{C}(\mathbf{C}_\mathrm{in}^a,\mathbf{C}_\mathrm{out}^b)
	+(1-\mu)\mathscr{C}(\mathbf{C}_\mathrm{in}^b,\mathbf{C}_\mathrm{out}^b),
\end{equation*}
where $\mathscr{C}$ denotes the average- or maximal-error capacity region
on both sides of the equation. 
\end{lem}
\begin{proof}
Our proof is via time-sharing. Consider two sequences of codes that achieve
the rate pairs
$(R_{1a},R_{2a})\in\mathscr{C}(\mathbf{C}_\mathrm{in}^a,\mathbf{C}_\mathrm{out}^a)$
and $(R_{1b},R_{2b})\in\mathscr{C}(\mathbf{C}_\mathrm{in}^b,\mathbf{C}_\mathrm{out}^b)$,
respectively.
Fix $\mu\in [0,1]$. Set $k=\lfloor n\mu\rfloor$ and $\ell=\lfloor n(1-\mu)\rfloor$.
Our aim is to show that concatenating the code with blocklength $k$ from the sequence achieving
$(R_{1a},R_{2a})$ and the code with blocklength $\ell$ from the sequence achieving 
$(R_{1b},R_{2b})$ results in a $(k+\ell)$-blocklength code for the MAC with 
a $(\mathbf{C}^\mu_\mathrm{in},\mathbf{C}^\mu_\mathrm{out})$-CF
that has small maximal or average error, depending on whether the original codes
have small maximal or average error. For $i\in \{1,2\}$, define the message
set of encoder $i$ as 
\begin{equation*}
  \mathcal{M}_i=[2^{kR_{ia}}]\times [2^{\ell R_{ib}}].
\end{equation*}
We denote the elements of $\mathcal{M}_i$ with $m_i=(m_{ia},m_{ib})$,
where $m_{ia}\in [2^{kR_{ia}}]$ and $m_{ib}\in [2^{\ell R_{ib}}]$.  
Note that 
\begin{equation*}
  \lim_{n\rightarrow\infty}\frac{1}{n}\log|\mathcal{M}_i|
	= \mu R_{ia}+(1-\mu) R_{ib}.
\end{equation*}
In addition, 
\begin{align*}
  \lfloor 2^{n(\mu C_\mathrm{in}^{ai}+(1-\mu)C_\mathrm{in}^{bi})}\rfloor
  &\geq \lfloor 2^{kC_\mathrm{in}^{ai}}\rfloor\times
  \lfloor 2^{\ell C_\mathrm{in}^{bi}}\rfloor\\
  \lfloor 2^{n(\mu C_\mathrm{out}^{ai}+(1-\mu)C_\mathrm{out}^{bi})}\rfloor
  &\geq \lfloor 2^{kC_\mathrm{out}^{ai}}\rfloor\times
  \lfloor 2^{\ell C_\mathrm{out}^{bi}}\rfloor
\end{align*}
Thus over the  $(\mathbf{C}_\mathrm{in}^\mu,\mathbf{C}_\mathrm{out}^\mu)$-CF, 
it is possible to transmit the concatenation of the symbols that 
our blocklength-$k$ and blocklength-$\ell$ codes transmit over the
$(\mathbf{C}_\mathrm{in}^a,\mathbf{C}_\mathrm{out}^a)$
and $(\mathbf{C}_\mathrm{in}^b,\mathbf{C}_\mathrm{out}^b)$-CFs, respectively.

Using the above construction, we see that the probability of error of the new code
when the message pair $(m_1,m_2)$ is transmitted, can be written as 
\begin{align*}
  \MoveEqLeft
  \pr\Big\{(\widehat{m}_1,\widehat{m}_2)\neq (m_1,m_2)\Big|\text{pair }(m_1,m_2)\text{ is transmitted}\Big\}\\
	&= \pr\Big\{(\widehat{m}_{1a},\widehat{m}_{2a})\neq (m_{1a},m_{2a})
	\text{ or }(\widehat{m}_{1b},\widehat{m}_{2b})\neq (m_{1b},m_{2b})\Big|m_1,m_2\Big\}\\
	&\leq \lambda_k^a(m_{1a},m_{2a})+\lambda_\ell^b(m_{1b},m_{2b}),
\end{align*}
where the last inequality follows from the union bound, and $\lambda_k^a$ and
$\lambda_\ell^b$ denote the probability of error of our original blocklength-$k$
and blocklength-$\ell$ codes when message pairs $(m_{1a},m_{2a})$ and $(m_{1b},m_{2b})$ are transmitted, respectively. Similarly, 
the average probability of error can be written as
\begin{align*}
  \MoveEqLeft
  \frac{1}{|\mathcal{M}_1||\mathcal{M}_2|}\sum_{m_1,m_2}\pr\Big\{(\widehat{m}_1,\widehat{m}_2)\neq (m_1,m_2)\Big|(m_1,m_2)\Big\}\\
	&\leq \frac{1}{|\mathcal{M}_1||\mathcal{M}_2|}\sum_{m_1,m_2}
	\Big(\lambda_k^a(m_{1a},m_{2a})+\lambda_\ell^b(m_{1b},m_{2b})\Big)\\
	&\leq P_{e,a}^{(k)}+P_{e,b}^{(\ell)},
\end{align*}
where $P_{e,a}^{(k)}$ and $P_{e,b}^{(\ell)}$ denote the average error probability
of the first and second code, respectively. 
\end{proof}

We next prove the theorem. Let 
$(R_{1a}^*,R_{2a}^*)\in\mathscr{C}(\mathbf{C}_\mathrm{in}^a,\mathbf{C}_\mathrm{out}^a)$
and $(R_{1b}^*,R_{2b}^*)\in\mathscr{C}(\mathbf{C}_\mathrm{in}^b,\mathbf{C}_\mathrm{out}^b)$
satisfy
\begin{align*}
  \alpha R_{1a}^*+(1-\alpha)R_{2a}^* &= C^\alpha(\mathbf{C}_\mathrm{in}^a,\mathbf{C}_\mathrm{out}^a)\\
	\alpha R_{1b}^*+(1-\alpha)R_{2b}^* &= C^\alpha(\mathbf{C}_\mathrm{in}^b,\mathbf{C}_\mathrm{out}^b).
\end{align*}
Then
\begin{align*}
  \MoveEqLeft
	\alpha \big(\mu R_{1a}^*+(1-\mu)R_{1b}^*\big)+
	(1-\alpha)\big(\mu R_{2a}^*+(1-\mu)R_{2b}^*\big)\\
	&=\mu C^\alpha(\mathbf{C}_\mathrm{in}^a,\mathbf{C}_\mathrm{out}^a)
	+(1-\mu)C^\alpha(\mathbf{C}_\mathrm{in}^b,\mathbf{C}_\mathrm{out}^b).
\end{align*}
Now since
\begin{equation*}
  \mu (R_{1a}^*,R_{2a}^*)+(1-\mu)(R_{1b}^*,R_{2b}^*)
\end{equation*}
is in $\mathscr{C}(\mathbf{C}_\mathrm{in}^\mu,\mathbf{C}_\mathrm{out}^\mu)$,
we have
\begin{equation*}
  C^\alpha(\mathbf{C}_\mathrm{in}^\mu,\mathbf{C}_\mathrm{out}^\mu)
	\geq \mu C^\alpha(\mathbf{C}_\mathrm{in}^a,\mathbf{C}_\mathrm{out}^a)
	+(1-\mu) C^\alpha(\mathbf{C}_\mathrm{in}^b,\mathbf{C}_\mathrm{out}^b).
\end{equation*}

Any convex (or concave) function defined on an open convex subset of
$\mathbb{R}^n$ is continuous \cite[pp. 22-23]{Lucchetti}. 
Thus for every $\alpha\in [0,1]$,
$C^\alpha$ is continuous on $\mathbb{R}^4_{>0}$. 

\subsection{Lemma \ref{lem:char} (Characterization of special regions in
\texorpdfstring{$\mathbb{R}^2_{\geq 0}$}{the plane})} \label{subsec:char}
Here we prove a generalization of Lemma \ref{lem:char} to arbitrary dimensions. Let $k$ be a positive integer 
and $\mathscr{C}$ be a compact subset of $\mathbb{R}^k_{\geq 0}$. In addition,
let $\Delta_k\subseteq\mathbb{R}^k_{\geq 0}$ denote the $k$-dimensional probability
simplex, that is, the set of all $\boldsymbol{\alpha}=(\alpha_1,\dots,\alpha_k)$
in $\mathbb{R}^k_{\geq 0}$ such that $\sum_{j=1}^k\alpha_j=1$. 
For every $\boldsymbol{\alpha}=(\alpha_1,\dots,\alpha_k)\in \Delta_k$, define
$C^{\boldsymbol{\alpha}}\in\mathbb{R}_{\geq 0}$ as
\begin{equation*}
  C^{\boldsymbol{\alpha}}
	=\max_{\boldsymbol{x}\in\mathscr{C}}
	\boldsymbol{\alpha}^T\boldsymbol{x}.
\end{equation*}
For $j\in [k]$, define the projection
$\pi_j:\mathbb{R}^k\rightarrow\mathbb{R}^k$ as
\begin{equation*}
  \pi_j(x_1,\dots,x_{j-1},x_j,x_{j+1},\dots,x_k)
	=(x_1,\dots,x_{j-1},0,x_{j+1},\dots,x_k).
\end{equation*}
In words, $\pi_j$ sets the $j$th coordinate of its input to zero and leaves
the other coordinates unchanged. We say a set $\mathscr{C}\subseteq \mathbb{R}^k$
is closed under $\pi_j$ if and only if $\pi_j(\mathscr{C})\subseteq \mathscr{C}$. 
\begin{lem}
Let $\mathscr{C}\subseteq\mathbb{R}^k_{\geq 0}$ be non-empty, compact, convex, 
and closed under the projections $\{\pi_j\}_{j=1}^k$. Then
\begin{equation} \label{eq:CalphaRegion}
  \mathscr{C}=\Big\{\mathbf{x}\in\mathbb{R}^k_{\geq 0}\Big|
	\forall \boldsymbol{\alpha}\in \Delta_k:
	\boldsymbol{\alpha}^T\boldsymbol{x}\leq 
	C^{\boldsymbol{\alpha}}\Big\}.
\end{equation}
\end{lem}

\begin{proof} Let $\mathscr{C}'$ denote the set on the right hand side
of (\ref{eq:CalphaRegion}). From the definition of $C^{\boldsymbol{\alpha}}$,
it follows $\mathscr{C}\subseteq\mathscr{C}'$. Thus it suffices to show
$\mathscr{C}'\subseteq\mathscr{C}$.

Every hyperplane in $\mathbb{R}^k$ divides $\mathbb{R}^k$ into two sets, each of which is
referred to as a half-space. Since $\mathscr{C}$ is closed and convex,
it equals the intersection of all the half-spaces containing it \cite[p. 36]{Boyd}. Thus
it suffices to show if for some $\boldsymbol{\beta}=(\beta_j)_{j=1}^k\in\mathbb{R}^k$ 
and $\gamma\in\mathbb{R}$ the half-space 
\begin{equation*}
 H=\Big\{\mathbf{x}\in\mathbb{R}^k\Big|\boldsymbol{\beta}^T\mathbf{x}\leq \gamma\Big\}
\end{equation*}
contains $\mathscr{C}$, then it also contains $\mathscr{C}'$. 
Suppose $H$ contains $\mathscr{C}$. 
Since $\mathscr{C}$ is nonempty and closed under the projections $\{\pi_j\}_{j=1}^k$,
$\mathscr{C}$ contains the origin. But $\mathscr{C}\subseteq H$, thus $H$ contains
the origin as well. This implies $\gamma\geq 0$. 

Let $S$ be the set of all $j\in [k]$ such that $\beta_j>0$.
If $S$ is empty, then $H$ contains $\mathbb{R}^k_{\geq 0}$ and by inclusion, $\mathscr{C}'$. 
Thus without loss of generality, we may assume $S$ is nonempty. In this case,
define $\boldsymbol{\alpha}=(\alpha_j)_{j\in [k]}\in \Delta_k$ as
\begin{equation*}
  \alpha_j =\begin{cases}
	\beta_j/\beta_S &\text{if }j\in S\\
	0 &\text{otherwise,}
	\end{cases}
\end{equation*}  
where $\beta_S=\sum_{j\in S}\beta_j>0$. From the definition of $C^{\boldsymbol{\alpha}}$,
it follows that there exists $\mathbf{x}\in\mathscr{C}$ such that
$\boldsymbol{\alpha}^T\mathbf{x}=C^{\boldsymbol{\alpha}}$, or equivalently,
\begin{equation} \label{eq:betaSC}
  \sum_{j\in S}\beta_j x_j=\beta_S C^{\boldsymbol{\alpha}}.
\end{equation}
Since $\mathscr{C}$ is closed under the projections $\{\pi_j\}_{j=1}^k$, the
vector $\mathbf{x}^*=(x^*_j)_{j\in [k]}$ is also in $\mathscr{C}$, where
\begin{equation*}
  x_j^* =\begin{cases}
	x_j &\text{if }j\in S\\
	0 &\text{otherwise.}
	\end{cases}
\end{equation*}
Using (\ref{eq:betaSC}) and the fact that
$\mathbf{x}^*\in \mathscr{C}\subseteq H$, we get
\begin{equation*}
  \beta_{S}C^{\boldsymbol{\alpha}}
	=\boldsymbol{\beta}^T\mathbf{x}^*\leq\gamma.
\end{equation*}
Now for every $\mathbf{x}'\in\mathscr{C}'$, we have 
\begin{equation*}
  \boldsymbol{\beta}^T \mathbf{x}'
	=\sum_{j=1}^k \beta_j x_j'
	\leq \sum_{j\in S}\beta_j x_j'
	=\beta_S \boldsymbol{\alpha}^T \mathbf{x}'
	\leq \beta_S C^{\boldsymbol{\alpha}}\leq \gamma.
\end{equation*}
Thus $\mathscr{C}'\subseteq H$. Since $H$ was an arbitrary half-space
containing $\mathscr{C}$, it follows $\mathscr{C}'\subseteq\mathscr{C}$.
\end{proof}

\subsection{Proposition \ref{prop:MaxNeqAvg} (Necessity of high capacity CF input links)} \label{subsec:MaxNeqAvg}

We show that for Dueck's contraction MAC \cite{Dueck},
there exists 
$\mathbf{C}_\mathrm{in}\in\mathbb{R}^2_{>0}$ such that 
for every $\mathbf{C}_\mathrm{out}\in\mathbb{R}^2_{\geq 0}$,  
$\mathscr{C}_\mathrm{max}(\mathbf{C}_\mathrm{in},\mathbf{C}_\mathrm{out})$
is a proper subset of 
$\mathscr{C}_\mathrm{avg}(\mathbf{C}_\mathrm{in},\mathbf{C}_\mathrm{out})$.
In Subsection \ref{subsec:contractionMAC}, we show that for the contraction MAC,
\begin{equation*}
  C^{1/2}_\mathrm{avg}(\mathbf{0},\mathbf{0})
	>C^{1/2}_\mathrm{max}(\mathbf{0},\mathbf{0}).
\end{equation*}
Thus it is possible to choose 
$\mathbf{C}_\mathrm{in}=(C^1_\mathrm{in},C^2_\mathrm{in})\in\mathbb{R}^2_{>0}$
such that 
\begin{equation*}
  C^{1/2}_\mathrm{avg}(\mathbf{0},\mathbf{0})
	-C^{1/2}_\mathrm{max}(\mathbf{0},\mathbf{0})>
	\frac{C_\mathrm{in}^1+C_\mathrm{in}^2}{2}.
\end{equation*}
For every $\mathbf{C}_\mathrm{out}\in\mathbb{R}^2_{\geq 0}$, we have
\begin{align*}
  C^{1/2}_\mathrm{max}(\mathbf{C}_\mathrm{in},\mathbf{C}_\mathrm{out})
  &\overset{(*)}{\leq} C^{1/2}_\mathrm{max}(\mathbf{0},\mathbf{0})
  + \frac{C_\mathrm{in}^1+C_\mathrm{in}^2}{2}\\
  &<C^{1/2}_\mathrm{avg}(\mathbf{0},\mathbf{0})\\
  &\leq C^{1/2}_\mathrm{avg}(\mathbf{C}_\mathrm{in},\mathbf{C}_\mathrm{out}),
\end{align*}
where $(*)$ follows from arguments similar to those that appear in 
the proof of Lemma \ref{prop:conf}. This completes the proof. 

\subsection{Proposition \ref{prop:discontinuity} 
(Discontinuity of \texorpdfstring{$C^\alpha$}{Calpha} under CF model)} 
\label{subsec:discontinuity}
Choose $\lambda\in (0,1)$ such that 
\begin{equation*}
  \min\{C_\mathrm{in}^1,C_\mathrm{in}^2\}
  >\lambda \max_{p(x_1,x_2)} I(X_1,X_2;Y),
\end{equation*}
and define $\mathbf{C}_\mathrm{in}^*=(C_\mathrm{in}^{*1},C_\mathrm{in}^{*2})$,
where $C_\mathrm{in}^{*i}=C_\mathrm{in}^i/\lambda$ for $i\in\{1,2\}$. 
Then
\begin{align*}
  \lim_{C_\mathrm{out}\rightarrow 0^+}
	C_\mathrm{max}^\alpha(\mathbf{C}_\mathrm{in}^*,(C_\mathrm{out},C_\mathrm{out}))
	&=\lim_{C_\mathrm{out}\rightarrow 0^+}
	C_\mathrm{avg}^\alpha(\mathbf{C}_\mathrm{in}^*,(C_\mathrm{out},C_\mathrm{out}))\\&\geq
	C_\mathrm{avg}^\alpha(\mathbf{0},\mathbf{0})>
	C_\mathrm{max}^\alpha(\mathbf{0},\mathbf{0}),
\end{align*}
where the equality follows by Theorem \ref{thm:MaxEqAvg}. This shows 
$C_\mathrm{max}^\alpha(\mathbf{C}_\mathrm{in}^*,\mathbf{C}_\mathrm{out})$ is not 
continuous. Now from Theorem \ref{lem:concavity}, it follows that
\begin{align*}
  C_\mathrm{max}^\alpha(\mathbf{C}_\mathrm{in},\mathbf{C}_\mathrm{out}) &\geq
  \lambda C_\mathrm{max}^\alpha(\mathbf{C}_\mathrm{in}^*,\mathbf{C}_\mathrm{out})
  + (1-\lambda)C_\mathrm{max}^\alpha(\mathbf{0},\mathbf{C}_\mathrm{out})\\&=
  \lambda C_\mathrm{max}^\alpha(\mathbf{C}_\mathrm{in}^*,\mathbf{C}_\mathrm{out})
  + (1-\lambda)C_\mathrm{max}^\alpha(\mathbf{0},\mathbf{0}),
\end{align*}
which can be rearranged as
\begin{equation*}
  C_\mathrm{max}^\alpha(\mathbf{C}_\mathrm{in},\mathbf{C}_\mathrm{out})-C_\mathrm{max}^\alpha(\mathbf{0},\mathbf{0})\geq
  \lambda \big(C_\mathrm{max}^\alpha(\mathbf{C}_\mathrm{in}^*,\mathbf{C}_\mathrm{out})
  -C_\mathrm{max}^\alpha(\mathbf{0},\mathbf{0})\big).
\end{equation*}
Since $\lambda>0$, the discontinuity of $C_\mathrm{max}^\alpha(\mathbf{C}_\mathrm{in}^*,\mathbf{C}_\mathrm{out})$
implies the discontinuity of $C_\mathrm{max}^\alpha(\mathbf{C}_\mathrm{in},\mathbf{C}_\mathrm{out})$.

\subsection{Corollary \ref{cor:discontinuity} (Dueck's Contraction MAC)} 
\label{subsec:contractionMAC}

Dueck's introduction of the ``Contraction MAC'' in \cite{Dueck} proves
the existence of multiterminal networks where the maximal-error capacity region
is a strict subset of the average-error capacity region. 
The input and output alphabets of the contraction
MAC are given by
\begin{align*}  
  \mathcal{X}_1 &= \{A,B,a,b\}\\
	\mathcal{X}_2 &= \{0,1\}\\
	\mathcal{Y} &= \{A,B,C,a,b,c\}\times\{0,1\}.
\end{align*}
The channel is deterministic and defined by the function 
$f:\mathcal{X}_1\times\mathcal{X}_2\rightarrow\mathcal{Y}$, where
\begin{align*}
  f(a,0) &= f(b,0) = (c,0)\\
	f(A,1) &= f(B,1) = (C,1),
\end{align*}
and $f(x_1,x_2)=(x_1,x_2)$ for all other $(x_1,x_2)$.
Dueck \cite{Dueck} shows that the maximal-error capacity region of this channel is contained
in the set of all rate pairs $(R_1,R_2)$ that satisfy
\begin{align*}
  R_1 &\leq \log 3-p\\
	R_2 &\leq h(p)
\end{align*}
for some $0\leq p\leq 1/2$, where $h(p)$ denotes the binary entropy function.
Thus for every $\alpha\in [0,1]$,
\begin{align*}
  C^\alpha_\mathrm{max}(\mathbf{0},\mathbf{0})
	&\leq \max_{p\in [0,1/2]}
	\Big[\alpha(\log 3-p)+(1-\alpha)h(p)\Big]\\
	&= \alpha(\log 3-1)+(1-\alpha)\log\big(1+2^\frac{\alpha}{1-\alpha}\big),
\end{align*}
where the maximum is achieved by 
\begin{equation*}
  p^* = \frac{1}{1+2^{\frac{\alpha}{1-\alpha}}}.
\end{equation*}

We next provide a lower bound for $C^\alpha_\mathrm{avg}(\mathbf{0})$ for the 
contraction MAC. From the average-error capacity region of the MAC \cite{Ahlswede1,Ahlswede2,Liao},
it follows that for $\alpha\in [0,1/2]$,
\begin{equation*}
  C^\alpha_\mathrm{avg}(\mathbf{0},\mathbf{0})=
  \max_{p(x_1)p(x_2)}\Big(
  \alpha I(X_1;Y)+(1-\alpha)I(X_2;Y|X_1)\Big)
\end{equation*}
and for $\alpha\in [1/2,1]$, 
\begin{equation*}
  C^\alpha_\mathrm{avg}(\mathbf{0},\mathbf{0})=
	\max_{p(x_1)p(x_2)}\Big(
	\alpha I(X_1;Y|X_2)+(1-\alpha)I(X_2;Y)\Big).
\end{equation*}
Since the contraction MAC is deterministic, the above equations 
simplify to 
\begin{equation} \label{eq:ZeroOneHalf}
  C^\alpha_\mathrm{avg}(\mathbf{0},\mathbf{0})=
	\max_{p(x_1)p(x_2)}\Big(
	\alpha H(Y)+(1-2\alpha)H(Y|X_1)\Big)
\end{equation}
and 
\begin{equation} \label{eq:OneHalfZero}
  C^\alpha_\mathrm{avg}(\mathbf{0},\mathbf{0})=
	\max_{p(x_1)p(x_2)}\Big(
	(1-\alpha) H(Y)+(2\alpha-1)H(Y|X_2)\Big)
\end{equation}
for $\alpha\in [0,1/2]$ and $\alpha\in [1/2,1]$, respectively. 
Let the input distribution of the first transmitter be given by
\begin{equation*}
  p_{X_1}(A)=p_A,p_{X_1}(B)=p_B,p_{X_1}(a)=p_a,p_{X_1}(b)=p_b,
\end{equation*}
and the input distribution of the second transmitter
be given by $p_{X_2}(1)=q$ and $p_{X_2}(0)=1-q$. In addition, let
$Y_1$ and $Y_2$ denote the components of $Y$ so that $Y=(Y_1,Y_2)$.
Note that $Y_2=X_2$. We have
\begin{align*}
  H(Y) &= H(Y_1,Y_2)\\
	&= H(Y_2)+H(Y_1|Y_2)\\
	&= h(q)+qH(p_a,p_b,p_A+p_B)
	+(1-q)H(p_A,p_B,p_a+p_b),
\end{align*}
where $h(q)$ denotes the binary entropy function
\begin{equation*}
  h(q)=q\log\frac{1}{q}+(1-q)\log\frac{1}{1-q}.
\end{equation*}
Furthermore,
\begin{align*}
  H(Y|X_1) &= H(Y_1,Y_2|X_1)\\
	&= H(Y_2|X_1)= h(q),
\end{align*}
and
\begin{align*}
  H(Y|X_2) &= H(Y_1,Y_2|X_2)\\
	&= H(Y_1|X_2)\\
	&= H(Y_1,X_2)-H(X_2)\\
	&= H(Y)-h(q). 
\end{align*}
From (\ref{eq:ZeroOneHalf})
and (\ref{eq:OneHalfZero}) it follows for all $\alpha\in [0,1]$,
\begin{align*}
  C_\mathrm{avg}^\alpha(\mathbf{0},\mathbf{0}) &\geq \alpha H(Y)+(1-2\alpha)H(q)\\
	&= (1-\alpha)h(q)+\alpha\big[ qH(p_a,p_b,p_A+p_B)
	+(1-q)H(p_A,p_B,p_a+p_b)\big].
\end{align*}
If we set $q=p^*$, $p_A=p_B=1/3$, 
and $p_a=p_b=1/6$, we get 
\begin{equation*}
  C_\mathrm{avg}^\alpha(\mathbf{0},\mathbf{0}) \geq 
	(1-\alpha)h(p^*)+\alpha(\log 3-p^*/3).
\end{equation*}
Recall that 
\begin{equation*}
  C_\mathrm{max}^\alpha(\mathbf{0},\mathbf{0}) \leq 
	(1-\alpha)h(p^*)+\alpha(\log 3-p^*).
\end{equation*}
Thus $C_\mathrm{avg}^\alpha(\mathbf{0},\mathbf{0}) >C_\mathrm{max}^\alpha(\mathbf{0},\mathbf{0})$, unless
$\alpha=0$ or $p^*=0$ (which occurs if and only if $\alpha=1$). 

\subsection{Proposition \ref{prop:confReliability} (Reliability benefit of conferencing)}
\label{subsec:confReliability}

Our proof is similar to the proof of Theorem \ref{thm:main}.
However, using results from Willems \cite{WillemsMAC}, we get
a stronger result than the one obtained by direct application 
of Theorem \ref{thm:main}.

For $r_1,r_2\geq 0$ and $C_{12},C_{21}\geq 0$, let 
\begin{equation*}
  \mathscr{C}_{\mathrm{conf},(r_1,r_2)}(C_{12},C_{21})
\end{equation*}
denote the $(r_1,r_2)$-error capacity region of a MAC with 
$(C_{12},C_{21})$-conferencing. Here we show that if 
$C_{12},C_{21}\geq 0$, then
\begin{equation} \label{eq:C12C21Rel}
  \mathscr{C}_{\mathrm{conf,avg}}(C_{12},C_{21}) \subseteq
  \mathscr{C}_{\mathrm{conf},(C_{12},C_{21})}(C_{12},C_{21}).
\end{equation}
Note that inclusion in the reverse direction, that is,
\begin{equation*}
  \mathscr{C}_{\mathrm{conf,avg}}(C_{12},C_{21}) \supseteq
  \mathscr{C}_{\mathrm{conf},(C_{12},C_{21})}(C_{12},C_{21}).
\end{equation*}
follows from definition; thus (\ref{eq:C12C21Rel})
is all that we require to prove equality. 

We now prove (\ref{eq:C12C21Rel}). For every blocklength $n$ and 
every pair of positive integers $(L_1,L_2)$, consider message sets
of the form 
\begin{equation*}
  \mathcal{M}_i=[K_i]\times [L_i]
	\text{ for }i\in\{1,2\}, 
\end{equation*}
where $K_1 = 2\lfloor 2^{nC_{12}}\rfloor$ and
$K_2 = 2\lfloor 2^{nC_{21}}\rfloor$.
We know from Willems \cite{WillemsMAC}, that a single conferencing round
achieves any rate pair in the average-error capacity region of the MAC 
with $(C_{12},C_{21})$-conferencing. Furthermore, in that single round it 
suffices for encoder 1 to send the first $nC_{12}$ bits of its message to
encoder 2 and for encoder 2 to send the first $nC_{21}$ bits of its message 
to encoder 1.  Thus if $(R_1,R_2)$ is a rate pair in the average-error capacity region,
then for all $\epsilon,\delta>0$ and all sufficently large $n$, there exist 
encoding functions of the form 
\begin{equation} \label{eq:confEnc}
  f_i:[K_1]\times [K_2]\times [L_i]\rightarrow \mathcal{X}_i^n
	\text{ for }i\in\{1,2\},
\end{equation}
and a decoder of the form 
$g:\mathcal{Y}^n\rightarrow \mathcal{M}_1\times \mathcal{M}_2$, with
\begin{equation*}
  \frac{1}{n}\log K_iL_i>R_i-\delta
	\text{ for }i\in\{1,2\},
\end{equation*}  
and average probability of error given by
\begin{equation*}
  P_{e,\mathrm{avg}}^{(n)}=\frac{1}{K_1K_2L_1L_2}\sum_{k_1,k_2}\sum_{\ell_1,\ell_2}
  \lambda_n((k_1,\ell_1),(k_2,\ell_2))\leq \epsilon.
\end{equation*}
Let $S$ be a subset of $[K_1]\times [K_2]$ with cardinality 
$|S|= K_1K_2/4$ containing the
$(k_1,k_2)$ pairs with the smallest values of 
\begin{equation*}
  \frac{1}{L_1L_2}\sum_{\ell_1,\ell_2}
  \lambda_n((k_1,\ell_1),(k_2,\ell_2)).
\end{equation*}
For $i\in\{1,2\}$, let $K'_i= K_i/2$.
Since $K'_1K'_2\leq |S|$, there exists an injective function
$\varphi:[K_1']\times [K_2']\rightarrow S$. Now consider the 
code defined by the encoders $(f'_1,f'_2)$, where for $i\in\{1,2\}$,
\begin{equation*}
  f'_i:[K_1']\times [K_2']\times [L_i]
  \rightarrow\mathcal{X}^n
\end{equation*}
maps $(k'_1,k'_2,\ell_i)$ to $f_i(\varphi(k'_1,k'_2),\ell_i)$,
and a decoder $g':\mathcal{Y}^n\rightarrow [K'_1]\times [K'_2]\times [L_1]\times [L_2]$
defined as
\begin{equation*}
  g'(y^n)=
  \begin{cases}
	\big(\varphi^{-1}(\hat{k}_1,\hat{k}_2),\hat{\ell}_1,\hat{\ell}_2\big)
	&\text{if }(\hat{k}_1,\hat{k}_2)\in\mathrm{range}(\varphi)\\
	(1,1,\hat{\ell}_1,\hat{\ell}_2)&\text{otherwise},
	\end{cases}
\end{equation*}
where $(\hat{k}_1,\hat{k}_2,\hat{\ell}_1,\hat{\ell}_2):=g(y^n)$.

Note that when the pair $\big((k'_1,\ell_1),(k'_2,\ell_2)\big)$
is transmitted using the new code, the probability of error 
equals $\lambda_n((k_1,\ell_1),(k_2,\ell_2))$, where 
\begin{equation*}
  (k_1,k_2):=\varphi(k'_1,k'_2).
\end{equation*}
Thus
\begin{equation*}
  P_e^{(n)}(C_{12},C_{21})\leq \max_{(k_1,k_2)\in S}\frac{1}{L_1,L_2}
  \sum_{\ell_1,\ell_2}\lambda_n((k_1,\ell_1),(k_2,\ell_2))\leq 
	\frac{4\epsilon}{3}.
\end{equation*} 
In addition, for $i\in\{1,2\}$ and sufficiently large $n$,
\begin{equation*}
  \frac{1}{n}\log K'_iL_i
  =\frac{1}{n}\log K_iL_i-\frac{1}{n}> R_i-2\delta. 
\end{equation*}
Thus $(R_1,R_2)$ is in the $(C_{12},C_{21})$-error capacity region. This completes the proof. 

\subsection{Proposition \ref{prop:conf} 
(Continuity of \texorpdfstring{$C^\alpha_\mathrm{conf}$}{Cconf})} 
\label{subsec:conf}
First note that the functions 
$\mathbf{C}_\mathrm{in}:\mathbb{R}^2_{\geq 0}\rightarrow \mathbb{R}^2_{\geq 0}$
and $\mathbf{C}_\mathrm{out}:\mathbb{R}^2_{\geq 0}\rightarrow \mathbb{R}^2_{\geq 0}$
defined by 
\begin{align*}
  \mathbf{C}_\mathrm{in}(C_{12},C_{21})
	&=(C_{12},C_{21})\\
  \mathbf{C}_\mathrm{out}(C_{12},C_{21})
	&=(C_{21},C_{12})
\end{align*}
are continuous. Thus from the definition of $C_\mathrm{conf}^\alpha$, 
given by (\ref{eq:CF2Conf}), and Theorem \ref{lem:concavity}
it follows that for every $\alpha\in [0,1]$,
$C^\alpha_\mathrm{conf}$ is continuous on $\mathbb{R}^2_{>0}$.
We next deal with the specfic results regarding $C^\alpha_\mathrm{conf,avg}$
and $C^\alpha_\mathrm{conf,max}$.

From \cite{WillemsMAC}, we know that the average-error capacity region
of the MAC with $(C_{12},C_{21})$-conferencing is given by the closure
of the set of all rate pairs $(R_1,R_2)$ that satisfy
\begin{align*}
  (R_1-C_{12})^+ &< I(X_1;Y|U,X_2)\\
  (R_2-C_{21})^+ &< I(X_2;Y|U,X_1)\\
  (R_1-C_{12})^++(R_2-C_{21})^+ &< I(X_1,X_2;Y|U)\\
  R_1+R_2 &< I(X_1,X_2;Y)
\end{align*}
for some distribution $p(u)p(x_1|u)p(x_2|u)$. Thus
whenever the rate pair $(R_1,R_2)$ is in the average-error capacity
region of a MAC with $(C_{12},C_{21})$-conferencing, then the rate pairs
\begin{align*} 
  &\big(R_1,(R_2-C_{21})^+\big)\\
	&\big((R_1-C_{12})^+,R_2\big)
\end{align*}
are achievable for the same MAC with $(C_{12},0)$- and $(0,C_{21})$-conferencing,
respectively.  From this result it follows that for every $\alpha\in [0,1]$,
\begin{align}
  C^\alpha_\mathrm{conf,avg}(C_{12},C_{21})
	&\leq C^\alpha_\mathrm{conf,avg}(C_{12},0)+(1-\alpha)C_{21}\label{eq:two2zero}\\
	C^\alpha_\mathrm{conf,avg}(C_{12},C_{21})
	&\leq C^\alpha_\mathrm{conf,avg}(0,C_{21})+\alpha C_{12}.\label{eq:one2zero}
\end{align}
Since $C^\alpha_\mathrm{conf,avg}(0,0)\leq C^\alpha_\mathrm{conf,avg}(C_{12},0)$,
if we now set $C_{21}=0$ in (\ref{eq:one2zero}), we get 
\begin{equation*}
  C^\alpha_\mathrm{conf,avg}(0,0)\leq
  C^\alpha_\mathrm{conf,avg}(C_{12},0)
  \leq C^\alpha_\mathrm{conf,avg}(0,0)+\alpha C_{12}.
\end{equation*}
Thus $C^\alpha_\mathrm{conf,avg}(C_{12},0)$ is continuous on $\mathbb{R}_{\geq 0}$.
Similarly, we can show that $C^\alpha_\mathrm{conf,avg}$
is continuous on $\mathbb{R}_{\geq 0}\times \{0\}$, by combining 
$C^\alpha_\mathrm{conf,avg}(C_{12},C_{21})
\geq C^\alpha_\mathrm{conf,avg}(C_{12},0)$ with
(\ref{eq:two2zero}).
The continuity on $\{0\}\times \mathbb{R}_{\geq 0}$
follows similarly, and thus $C^\alpha_\mathrm{conf,avg}$ 
is continuous on $\mathbb{R}^2_{\geq 0}$. 

We next prove $C_\mathrm{conf,max}^{1/2}$, viewed as a function over $\mathbb{R}^2_{\geq 0}$,
is continuous at $(0,0)$. Note that for every $(n,M_1,M_2,J)$-code for the MAC with 
conferencing, the set of all messages that lead to the same conferencing 
output is of the form $\mathcal{A}_1\times \mathcal{A}_2$ for some $\mathcal{A}_i\subseteq [M_i]$
for $i\in \{1,2\}$. This follows directly from Equation (19) in \cite{WillemsMAC}.
Now fix a sequence of $(n,M_1,M_2,J)$-codes that achieve the rate pair $(R_1^*,R_2^*)$,
where
\begin{equation*}
  R_1^*+R_2^*
	=2C^{1/2}_\mathrm{max}(C_{12},C_{21}).
\end{equation*}
Since there are at most $2^{n(C_{12}+C_{21})}$ possible conferencing outputs, the
pigeonhole principle implies that for each $i\in\{1,2\}$, there exists a set 
$\mathcal{A}_i^*\subseteq [M_i]$ such that
\begin{equation*}
  |\mathcal{A}_1^*|\times |\mathcal{A}_2^*|
  \geq M_1 M_2 2^{-n(C_{12}+C_{21})},
\end{equation*}
and the set of all message pairs in $\mathcal{A}_1^*\times\mathcal{A}_2^*$
lead to the same conferencing output. 
Since for $i\in\{1,2\}$ and sufficiently large $n$, 
$\frac{1}{n}\log M_i\geq R_i^*-\delta$, we get
\begin{align*}
  \frac{1}{n}\log\big|\mathcal{A}_1^*\big|\big|\mathcal{A}_2^*\big|
  &\geq R_1^*+R_2^*-C_{12}-C_{21}-2\delta\\
  &= 2C^{1/2}_\mathrm{max}(C_{12},C_{21})-C_{12}-C_{21}-2\delta.
\end{align*}
Now consider the code where for $i\in\{1,2\}$, encoder $i$ transmits codewords from the 
original code that correspond to messages in $\mathcal{A}^*_i$. Then this code has 
small maximal error since the maximum probability of error over the message pairs 
in $\mathcal{A}_1^*\times\mathcal{A}_2^*$ is at most as large as the maximal probability
of error of the original code. Thus 
\begin{equation*}
  C^{1/2}_\mathrm{max}(C_{12},C_{21})\leq
  C^{1/2}_\mathrm{max}(0,0)+\frac{C_{12}+C_{21}}{2}.
\end{equation*}
Combining this inequality with 
\begin{equation*}
  C^{1/2}_\mathrm{max}(C_{12},C_{21})
	\geq C^{1/2}_\mathrm{max}(0,0),
\end{equation*}
implies $C^{1/2}_\mathrm{max}$ is continuous at $(0,0)$.

\section{Conclusion}

Cooperation is a powerful tool in communication networks. In addition to 
increasing transmission rates, cooperation makes communication
more reliable. Specifically, Theorem \ref{thm:main} and Proposition \ref{prop:confReliability} 
quantify the relationship between the reliability of a network and cooperation rate
under the CF and conferencing models, respectively. Theorem \ref{thm:MaxEqAvg}
states that in the CF model, when the facilitator has full
access to the messages, the maximal- and average-error capacity regions of the network
are identical, no matter how small the output link capacities of the CF are.
This result continues to hold even when the CF output links are of negligible capacity,
thus providing a positive answer to the question posed in the title of the paper. 
Finally, Proposition \ref{prop:discontinuity}
demonstrates the existence of a network whose maximal-error sum-capacity is not 
continuous with respect to the capacities of some of its edges. 
The same question, with average-error replacing maximal-error, remains open. 

\bibliographystyle{IEEEtran}
\bibliography{ref}{}

\end{document}